\documentclass[10pt,journal, final]{IEEEtran}
\usepackage{ifpdf}
\usepackage{cite}
\usepackage{algorithmic}
\usepackage{algorithm}
\ifCLASSINFOpdf
\usepackage[pdftex]{graphicx}

\else
 \usepackage[dvips]{graphicx}
\fi
\usepackage[cmex10]{amsmath}
\usepackage{epstopdf}
\usepackage{array}
\usepackage{flushend}
\usepackage{balance}
\usepackage{eqparbox}

\usepackage[tight,footnotesize]{subfigure}
\usepackage{fixltx2e}
\usepackage{color}
\usepackage{float}

\usepackage{dblfloatfix}
\usepackage{url}
\usepackage{cite}
\usepackage{subfigure}
\usepackage{amsmath}
\usepackage{epstopdf}\newtheorem{theorem}{Theorem}
\newtheorem{proposition}{Proposition}
\newenvironment{proof}{{\indent \indent \it Proof:}}{\hfill $\blacksquare$}
\newtheorem{remark}{Remark}
\newtheorem{lemma}{Lemma}
\newtheorem{assumption}{Assumption}
\usepackage{amssymb}
\usepackage{color}

\begin{document}
\allowdisplaybreaks[3]
\title{Interference Management in Reverse TDD HAPS-Ground Networks: A Large System Analysis}
\setlength{\columnsep}{0.21 in}

\author{Shasha~Liu,~\IEEEmembership{Student Member,~IEEE,}
Abla~Kammoun,~\IEEEmembership{Member,~IEEE,}      
and
Mohamed-Slim~Alouini,~\IEEEmembership{Fellow,~IEEE}


}
\maketitle
\begin{abstract}
High-altitude platform stations (HAPSs) are expected to coexist with terrestrial networks (TNs) to provide seamless wide-area wireless connectivity in future sixth-generation (6G) systems. However, the strong line-of-sight (LoS) propagation of HAPS links introduces severe cross-tier interference under spectrum sharing. This paper proposes a reverse time division duplexing (RTDD)-based transmission framework for integrated HAPS--TN systems, where the HAPS and terrestrial base station (BS) operate in opposite transmission directions to mitigate cross-tier interference while preserving TDD channel reciprocity. Standardized 3GPP channel models are adopted for both HAPS and terrestrial links. To enable tractable performance analysis, an analytical Kronecker channel model is developed to approximate the empirical 3GPP terrestrial BS channel, based on which a random matrix theory (RMT)-based large-system analysis is established to derive deterministic approximations for the uplink and downlink signal-to-interference-plus-noise ratios (SINRs). Numerical results demonstrate an excellent agreement between the proposed analysis and Monte Carlo simulations. Moreover, the proposed RTDD framework effectively suppresses cross-tier interference, achieves near interference-free performance.
\end{abstract}
\IEEEpeerreviewmaketitle
\section{Introduction}
High-altitude platform stations (HAPSs) have emerged as a promising technology for 6G and beyond, providing wide-area wireless connectivity from the stratosphere \cite{kurt2021vision,mohammed2011role,shibata2020system}. Owing to their high-altitude operation and deployment flexibility, HAPSs are well suited for extending wireless connectivity to remote, rural, mountainous, and disaster-stricken areas where terrestrial infrastructure is unavailable or difficult to deploy \cite{karaman2025solutions}. This capability is often referred to as connecting the unconnected.
In addition to extending wireless coverage, HAPSs are expected to complement existing terrestrial networks by providing additional capacity in temporary traffic hotspots and filling coverage gaps caused by terrain, buildings, or insufficient infrastructure, thereby connecting the connected. As a result, future wireless systems are envisioned to operate as integrated HAPS–terrestrial networks, where HAPSs and ground base stations (BSs) coexist and jointly provide seamless wireless services \cite{deng2026distributed,alam2021high}.
\par
One of the key advantages of HAPS communications is the strong line-of-sight (LoS) propagation between HAPS and ground users, which significantly enhances coverage and communication reliability. To fully exploit these advantages in future wireless systems, HAPSs are expected to coexist and share spectrum with terrestrial networks (TNs), particularly in densely populated regions where spectrum resources are scarce. However, the favorable LoS propagation of HAPS links also leads to strong cross-tier interference between HAPSs and terrestrial BSs under spectrum sharing, making efficient interference management a key challenge for enabling reliable and spectrum-efficient integrated HAPS--TN systems. The coexistence issue has attracted increasing research attention. For example, \cite{yuan2022interference} analyzed the cross-tier interference between HAPS and terrestrial mobile networks under co-channel deployment, highlighting the challenges of spectrum sharing.
\par
To mitigate cross-tier interference, massive multiple-input multiple-output (mMIMO) has emerged as a key enabling technology for integrated HAPS--TN systems \cite{bjornson2017massive,bjornson2025enabling, bjornson2015optimal}. By exploiting the large spatial degrees of freedom provided by massive antenna arrays, mMIMO enables highly directional beamforming that enhances the desired signals while effectively suppressing cross-tier interference. Motivated by these advantages, numerous studies \cite{liu2026sum, shamsabadi2024enhancing,jang2025haps, deng2025two} have investigated optimization-based beamforming schemes for integrated HAPS systems. Although these approaches can effectively mitigate interference, they generally rely on iterative optimization algorithms and require high computational complexity as well as instantaneous channel state information (CSI), making real-time implementation challenging in large-scale HAPS deployments.
\par 
An alternative approach is to exploit reverse time division duplexing (RTDD) to mitigate cross-tier interference at the network level. Reference \cite{hoydis2013making} showed that the combination of massive MIMO and RTDD can effectively suppress cross-tier interference in heterogeneous cellular networks. Building upon this idea, \cite{sanguinetti2015interference} employed large-system analysis to characterize interference management in reverse-TDD heterogeneous networks with wireless backhaul. In addition, \cite{xia2017large} further developed a random matrix theory (RMT)-based framework for the joint optimization of time and frequency resource allocation in large-scale RTDD heterogeneous networks with wireless backhaul. More recently, \cite{guo2018dynamic} further demonstrated the effectiveness of dynamic TDD for mitigating cross-link interference in cellular networks, highlighting its potential as an efficient spectrum-sharing mechanism.
For HAPS--terrestrial coexistence, \cite{fujii2024mobile} proposed an RTDD-based spectrum-sharing approach, where HAPS and terrestrial systems operate over the same frequency band with opposite uplink and downlink transmission timings to mitigate cross-tier interference. More recently, \cite{nakazawa2026adaptive} investigated adaptive control of the reverse use of the TDD frame structure to maximize the system throughput in co-existing HAPS and terrestrial BS networks.
\par
In addition to mitigating interference, RTDD preserves the channel reciprocity of conventional TDD because each network tier still operates in TDD mode. Consequently, accurate downlink CSI can be acquired from uplink pilots without explicit CSI feedback, making the channel estimation overhead independent of the number of antennas \cite{marzetta2010noncooperative,marzetta2006much}. Furthermore, unlike coordinated beamforming schemes that require exchanging instantaneous CSI between the HAPS and terrestrial BS, RTDD enables implicit coordination through transmission scheduling, thereby significantly reducing signaling overhead.
\par
Despite these advantages, the application of RTDD to integrated HAPS--TNs remains largely unexplored. Existing RTDD studies mainly focus on terrestrial heterogeneous networks and do not account for the unique propagation characteristics of HAPS links. Moreover, the lack of tractable analytical models makes it difficult to evaluate the performance of large-scale HAPS systems without relying on computationally expensive Monte Carlo simulations.
Motivated by these observations, this paper investigates RTDD for integrated HAPS--TN systems under practical 3GPP channel models and develops a tractable analytical framework to characterize its performance. The main contributions of this paper are summarized as follows.
\begin{itemize}
    \item We propose an integrated HAPS--TN transmission framework operating in the RTDD mode, where the HAPS and terrestrial BS operate in opposite transmission directions to effectively mitigate cross-tier interference while preserving the channel reciprocity of TDD.
    \item We establish a comprehensive evaluation framework by adopting the standardized 3GPP channel models for both the HAPS and terrestrial BS links. In the proposed RTDD framework, both the HAPS and the terrestrial BS employ zero-forcing (ZF) precoding for downlink transmission and minimum mean-square-error (MMSE) receivers for uplink signal detection. The proposed RTDD framework is benchmarked against the conventional TDD scheme, a null-space-aided TDD (NSA-TDD) scheme, and two interference-free reference configurations, namely the HAPS-only and BS-only systems.
    \item To facilitate tractable performance analysis, we derive an analytical Kronecker channel model that accurately approximates the empirical 3GPP terrestrial BS channel. Based on this model, we develop a large-system analytical framework using random matrix theory (RMT) to derive deterministic approximations for the uplink SINR at the BS and the downlink SINR of the BUEs under Rician fading channels, while the SINRs of the HAPS links are obtained in closed form by exploiting the deterministic LoS propagation.
    \item Numerical results verify the accuracy of the proposed deterministic approximations through an excellent agreement with Monte Carlo simulations. Moreover, the proposed RTDD framework effectively suppresses cross-tier interference, achieves near interference-free performance, and provides a low-complexity alternative to optimization-based interference mitigation schemes.
\end{itemize}

The remainder of this paper is organized as follows. Section~\ref{section:sytem model} introduces the integrated HAPS--TN system model and the proposed RTDD transmission protocol. Section~\ref{section:Large System Analysis} presents the large-system analysis and derives deterministic approximations for the uplink and downlink SINRs. Section~\ref{section:TDD} presents the conventional TDD scheme and the NSA-TDD benchmark for comparison.
Section~\ref{section:Simulation Result} provides numerical results to validate the proposed analytical framework and evaluate the performance of the RTDD scheme under practical 3GPP channel models. Finally, Section~\ref{section:Conclusions} concludes the paper.

\noindent{\bf Notations.}
The following notations are used. Boldface upper-case letters denote matrices, and boldface lowercase letters denote column vectors. The norm $\|.\|$ stands for the Euclidean norm for vectors and the associated operator norm for matrices. The notations $\mathbf{X}^T$,  $\mathbf{X}^H$ denote the transpose and conjugate of $\mathbf{X}$, respectively.  The trace of the matrix $\mathbf{X}$ is denoted by $\operatorname{tr}(\mathbf{X}) $. 
We use $[\cdot]_{i, k}$ to denote the $(i, k)^{th}$ element of
the enclosed matrix. We write \({\bf x} \sim \mathcal{CN}(\boldsymbol{\mu}, \boldsymbol{\Sigma})\) to denote a circularly symmetric complex Gaussian vector with mean \(\boldsymbol{\mu}\) and covariance \(\boldsymbol{\Sigma}\).
The notation $\mathbf{I}_M$ is the $M \times M$ identity matrix.

\section{System Model}
\label{section:sytem model}
\begin{figure*}[!h]
\centering
\subfigure[HAPS DL--BS UL phase: one spatial DoF at the HAPS is used to suppress interference toward the BS.]{
\includegraphics[width=3in]{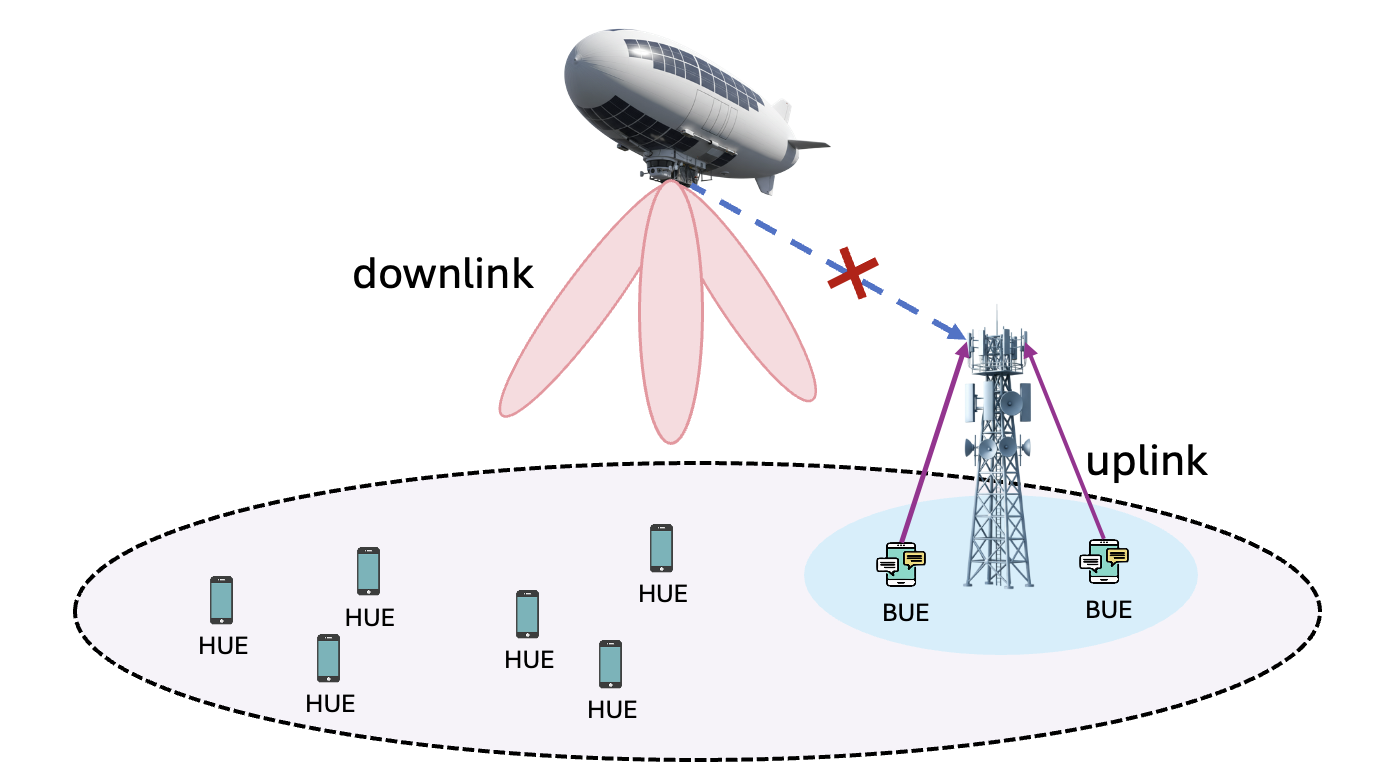}
\label{Fig:SM_HAPS_DL}
}
\hspace{0.3in}
\subfigure[HAPS DL--BS UL phase: the interference from the BS to the HAPS is negligible due to the down-tilted BS antennas.]{
\includegraphics[width=3in]{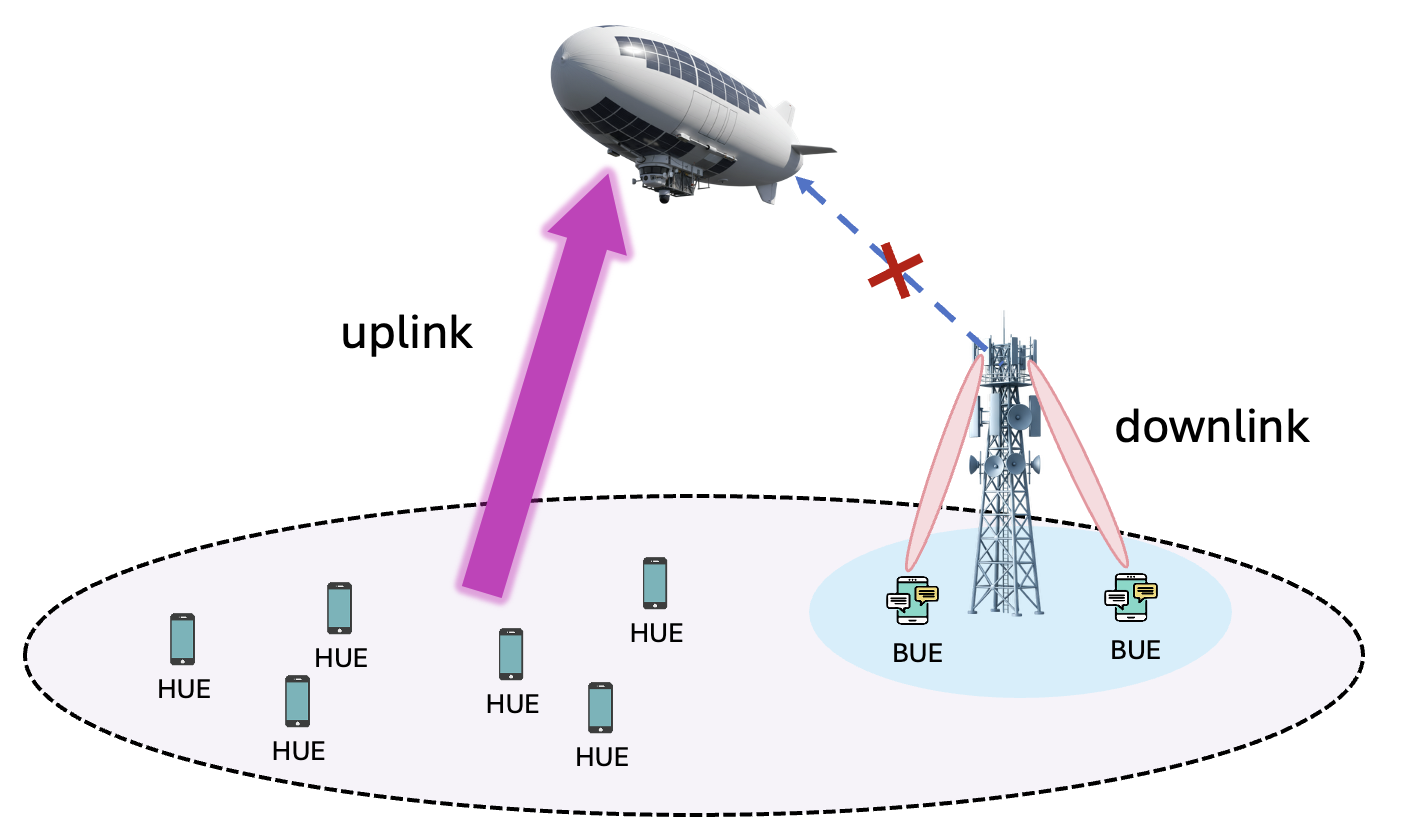}
\label{Fig:SM_HAPS_UL}
}
\caption{Illustration of the two transmission phases in the considered RTDD HAPS--ground network.}
\label{Fig:SM}
\end{figure*}
As shown in Fig.~\ref{Fig:SM}, we consider an integrated HAPS--TN consisting of one HAPS and one terrestrial BS. The HAPS is equipped with $M$ antennas and serves  $K$ single-antenna HAPS-served user equipments (HUEs). The HUEs are uniformly distributed within the HAPS coverage area excluding the BS coverage region. The terrestrial BS is equipped with $N$ antennas and serves  $U$ single-antenna BS-served user equipments (BUEs).

\subsection{Transmission Protocol}
The proposed transmission protocol is illustrated in Fig.~\ref{Fig:trans_protocol}. We adopt a co-channel RTDD protocol \cite{hoydis2013making}, where the HAPS and terrestrial BS operate in opposite transmission directions over the same time-frequency resources.
During the first transmission interval $T_1$, the HAPS operates in the downlink (DL) while the terrestrial BS operates in the uplink (UL). During the second transmission interval $T_2$, the transmission directions are reversed, i.e., the HAPS operates in the UL and the BS operates in the DL.
\begin{figure}[h]
\centering
\includegraphics[width=3.5in]{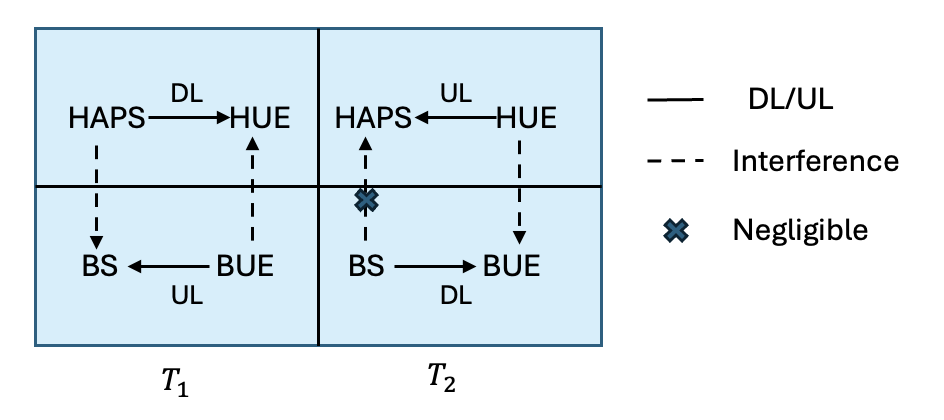}
\caption{Illustration of the transmission protocol. The exchange of information within each tier takes place in reverse order, i.e., the HAPS is in the DL mode (HAPS → HUE) when the BS operates in the UL (BS ← BUE), and vice versa.}
\label{Fig:trans_protocol}
\end{figure}
\par
The advantages of the previously described co-channel RTDD transmission protocol are listed as follows.
\begin{itemize}
    \item In the first time interval $T_1$, co-channel interference from the HAPS downlink transmission to the BS receiver can be avoided
    by constraining the HAPS precoder to lie in the null space of the channel between the HAPS antennas and the BS.
    \item In the second time interval $T_2$, the HAPS experiences co-channel interference from the BS, which can be negligible since the BS antennas are down-tilted toward the ground, while the HAPS is located at a high elevation angle and therefore lies outside the BS main radiation lobe.
\end{itemize}

\subsection{Channel Model}
In this subsection, we introduce the channel models considered in the system.

\subsubsection{HAPS-to-UE Channels} 
Owing to the dominant line-of-sight (LoS) propagation between the HAPS and ground users $l=1,\cdots,K$, the HAPS-to-UE $l$ channel is modeled as
\begin{equation}
\label{Eq:channel_HAPS}
\mathbf{h}_l^{(\rm HAPS)} = \sqrt{\beta_l^{(\rm HAPS)}}\mathbf h_{l,\mathrm{HAPS}}^{\rm LoS},
\end{equation}
where \(\beta_l^{(\rm HAPS)}\) denotes the large-scale channel gain. The large-scale channel gain expressed in dB is given by
$
\beta_{l,\mathrm{dB}}^{(\rm HAPS)} = G_{\mathrm{UE}} - PL_{l,\mathrm{dB}} + G_{E,\mathrm{max}},
$
where \(G_{\rm UE}\), \(PL_{l,\rm dB}\), and \(G_{E,\max}\) denote the UE antenna gain, free-space path loss, and maximum antenna element gain, respectively.
\par
The LoS channel vector is expressed element-wise as
\begin{equation}
[\mathbf h_{l,\mathrm{HAPS}}^{\rm LoS}]_p 
\!\!=\!\! e^{-j 2\pi d_{3D,l}/\lambda}\,
\sqrt{F_p(\phi_l^{\mathrm{LoS}},\theta_l^{\mathrm{LoS}})}\,[\mathbf{a}_{\mathrm{tx}}^H(\phi_l^{\mathrm{LoS}},\theta_l^{\mathrm{LoS}})]_p,
\end{equation}
where $d_{3D,l}$ is the 3D link distance, $\lambda$ is the wavelength, $F_p(\cdot)$ denotes the antenna 
element pattern, and $\mathbf{a}_{\mathrm{tx}}(\cdot)$ is the transmit array response, while $\phi_l^{\mathrm{LoS}}$ and $\theta_l^{\mathrm{LoS}}$ are the azimuth and zenith angle for LoS direction.
\subsubsection{BS-to-UE Channels}
The BS-to-UE channel is modeled as a Rician fading model consisting of deterministic LoS and non-line-of-sight (NLoS) components, i.e.,
\begin{equation}
\label{Eq:channel}
\mathbf{h}_l^{(\rm BS)}
=
\sqrt{\beta_l^{\rm BS}}
\left(
\sqrt{\frac{\kappa_l}{\kappa_l+1}}
\mathbf{h}_l^{\mathrm{LoS}}
+
\sqrt{\frac{1}{\kappa_l+1}}
\mathbf{h}_l^{\mathrm{NLoS}}
\right),
\end{equation}
where \(\beta_l^{\rm BS}\) and \(\kappa_l\) denote the path-loss and the Rician \(K\)-factor associated with UE \(l\), respectively.

The LoS component is given by
\begin{equation}
\mathbf{h}_l^{\mathrm{LoS}}
=
\exp\!\left(-j2\pi\frac{d_{3D,l}}{\lambda}\right)
\sqrt{g_E\!\left(\phi_l^{\rm LoS},\theta_l^{\rm LoS}\right)}
\mathbf{a}_{\rm tx}^{H}
\!\left(\phi_l^{\rm LoS},\theta_l^{\rm LoS}\right),
\end{equation}
The NLoS component is modeled in two different ways.
The first model directly adopts the standardized 3GPP stochastic channel model specified in TR~38.901 \cite{3gpp_tr38901},
\begin{equation}
\label{Eq:NLoS_3GPP}
\mathbf{h}_l^{\mathrm{NLoS}}
=
\mathbf{h}_{\mathrm{3gpp},l}^{\mathrm{NLoS}},
\end{equation}
where \(\mathbf{h}_{\mathrm{3gpp},l}^{\mathrm{NLoS}}\) is generated according to the 3GPP channel generation procedure.
\par
To facilitate analytical performance analysis, we further construct a Kronecker channel model whose spatial correlation is extracted from the empirical 3GPP channel realizations. Specifically, the correlation matrix is estimated as
\begin{equation}
\label{Eq:R_3gpp}
\mathbf{R}_{l}^{\mathrm{3gpp}}
=
\frac{1}{N_s}
\sum_{i=1}^{N_s}
\mathbf{h}_{\mathrm{3gpp},l}^{\mathrm{NLoS}}(i)
\left(
\mathbf{h}_{\mathrm{3gpp},l}^{\mathrm{NLoS}}(i)
\right)^H.
\end{equation}
where $N_s$ is the number of independent realizations of \(\mathbf{h}_{\mathrm{3gpp},l}^{\mathrm{NLoS}}\).
\par
The resulting Kronecker channel model is
\begin{equation}
\label{Eq:NLoS_Kronecker}
\mathbf{h}_l^{\mathrm{NLoS}}
=
\left(
\mathbf{R}_{l}^{\mathrm{3gpp}}
\right)^{1/2}
\mathbf{z}_l,
\end{equation}
where \(\mathbf{z}_l \sim \mathcal{CN}(\mathbf{0},\mathbf{I})\).
This model preserves the spatial correlation characteristics of the 3GPP channel while providing a tractable analytical representation for the subsequent RMT analysis.

\subsubsection{HUEs to BUEs}
The interference channel from HUE \(k\) to BUE \(u\) is modeled as a Rayleigh fading channel,
\begin{equation}
h_{ku}=\sqrt{\beta_{ku}}z_{ku}
\end{equation}
where \(\beta_{ku}\) denotes the large-scale fading coefficient accounting for path loss and penetration loss, and \(z_{ku}\sim\mathcal{CN}(0,1)\).


\section{Large System Analysis for RTDD}
\label{section:Large System Analysis}
This section jointly introduces the beamforming vectors employed reverse TDD and analyzes their resulting performance. The derivations are presented separately for the two transmission intervals, with the beamforming design and the corresponding SINR analysis developed together for each case. Since the HAPS links are dominated by deterministic LoS propagation, their uplink and downlink SINRs admit closed-form expressions. In contrast, the terrestrial BS links experience Rician fading, making the corresponding SINRs random quantities. The proposed framework can be evaluated using either the empirical 3GPP channel model or the analytical Kronecker channel model. Unless otherwise specified, all RMT-based analysis in this section relies on the Kronecker channel model.

 \subsection{Analysis for Transmission Interval $T_1$}
During the first transmission interval $T_1$, the HAPS operates in the DL and serves the HUEs, while the terrestrial BS operates in the UL and receives signals from the BUEs. The performance of the two links is analyzed separately in the following.
\subsubsection{HAPS Downlink}
The received signal at HUE $k=1,\cdots, K$ is given by
\begin{equation}
\begin{aligned}
y_k^{(\rm HUE)} 
=&\sum_{i\in \mathcal{K}}
\sqrt{\frac{P^{(\rm HAPS)}}{K}} (\mathbf{h}_k^{(\rm HAPS)})^H \mathbf{w}_i^{(\rm HAPS)} x_i^{(\rm HAPS)}\\
+& \sum_{u\in \mathcal{U}} \sqrt{p_u} h_{ku} x_u + n_k ,
\end{aligned}
\end{equation}
where $\mathbf{h}_k^{(\rm HAPS)} \in \mathbb{C}^{M\times 1}$ denotes the channel vector from the HAPS to HUE $k$, and $\mathbf{w}_k^{(\rm HAPS)} \in \mathbb{C}^{M\times 1}$ is the corresponding precoding vector, whereas the scalar $h_{ku}$ represents the channel coefficient from BUE $u$ to HUE $k$, 
$P^{(\mathrm{HAPS})}$ denotes the transmit power of HAPS, and $p_u$ denotes the transmit power of BUE $u$. The transmitted symbols satisfy $x_i^{(\rm HAPS)} \sim \mathcal{CN}(0,1)$ and $x_u \sim \mathcal{CN}(0,1)$. The additive white gaussian noise is modeled as $n_k \sim \mathcal{CN}(0,\sigma^2)$.
\par
Let $\mathbf{W}^{(\rm HAPS)}=[\mathbf{w}_1^{(\rm HAPS)},\mathbf{w}_2^{(\rm HAPS)}, \cdots, \mathbf{w}_K^{(\rm HAPS)}]\in \mathbb{C}^{M\times K}$ be the precoding matrix in HAPS. The SINR at HUE $k=1,\cdots,K$ is given as:
\begin{equation}
\label{eq:HAPS_DL_SINR}
\begin{aligned}
&\rm{SINR}_k^{(\rm HUE,DL)}=\\
&\frac{\frac{P^{(\rm HAPS)}}{K}|(\hat{\mathbf{h}}_k^{(\rm HAPS)})^H\mathbf{w}_k^{(\rm HAPS)}|^2}{\sum_{\substack{i \in \mathcal{K} \\ i\neq k}}\frac{P^{(\rm HAPS)}}{K}|(\hat{\mathbf{h}}_k^{(\rm HAPS)})^H\mathbf{w}_i^{(\rm HAPS)}|^2+\sum_{u \in \mathcal{U}}p_u|\hat{h}_{ku}|^2+1}
\end{aligned}
\end{equation}
where $\hat{\mathbf{h}}_k^{(\rm HAPS)}= \sqrt{\rho_k^{(\rm HAPS)}}\mathbf{h}_{k,\rm HAPS}^{\mathrm{LoS}}$ with $\rho_k^{(\rm HAPS)}=\frac{\beta_k^{(\rm HAPS)}}{\sigma^2}$, $\hat{h}_{ku}=\sqrt{\rho_{ku}}z_{ku}$, with $\rho_{ku}=\frac{\beta_{ku}}{\sigma^2}$. 
\par
Considering that the strong LoS link between the HAPS and the BS leads to severe HAPS-to-BS interference,  the HAPS adopts linear precoding and sacrifices one spatial degree of freedom to null its interference toward the BS while simultaneously serving the HUEs.
Because the HAPS--BS channel is LoS-dominated, the channel matrix is of rank one and can be represented by the vector $\mathbf h_{HB}$. The corresponding projection matrix is
\begin{equation}
\mathbf{T} =
\mathbf{I}_M -
\frac{\mathbf{h}_{HB}\mathbf{h}_{HB}^H}{\mathbf{h}_{HB}^H\mathbf{h}_{HB}}.
\label{eq:T_projection}
\end{equation}
The effective channel after null-space projection is $\mathbf{U}=\mathbf{T}\hat{\mathbf{H}}^{(\rm HAPS)}_{\rm HUE}\in \mathbb{C}^{M\times K}$, where   $\hat{\mathbf{H}}^{(\rm HAPS)}_{\rm HUE}=[\hat{\mathbf{h}}_1^{(\rm HAPS)},\hat{\mathbf{h}}_2^{(\rm HAPS)},\cdots, \hat{\mathbf{h}}_K^{(\rm HAPS)}]$.
To eliminate the intra-cell interference among the HUEs, ZF precoding is adopted. The ZF precoder is constructed as
$$
\widetilde{\mathbf W}^{(\rm HAPS)}=\mathbf{U}(\mathbf{U}^H\mathbf{U})^{-1},
$$
The normalized transmit beamforming matrix is then given by
$
\mathbf W^{(\rm HAPS)}
= \widetilde{\mathbf W}^{(\rm HAPS)} \mathbf \Lambda,
$
where
$
\mathbf\Lambda
=\operatorname{diag}
\left(
\frac1{\|\widetilde{\mathbf w}_1^{(\rm HAPS)}\|},
\ldots,
\frac1{\|\widetilde{\mathbf w}_K^{(\rm HAPS)}\|}
\right),
$
\par
\begin{assumption}
\label{assump:HUE_interference}
 There exists $\rho_{\rm min}$ and $\rho_{\rm max}$ such that the normalized large-scale fading coefficients satisfy
\begin{equation}
\frac{\rho_{\rm min}}{U}\leq \min_{\substack{u=1,\cdots,U\\ k=1,\cdots,K}}\rho_{ku}\leq\max_{\substack{u=1,\cdots,U\\ k=1,\cdots,K}}\rho_{ku}<\frac{\rho_{\rm max}}{U}
\end{equation}
and there exists a constant $P_{\max}>0$ such that
\begin{equation}
0<p_i\le P_{\max},\qquad \forall i \in \mathcal{U}\cup\mathcal{K}.
\end{equation}
\end{assumption}
\begin{remark}
\label{remark:1}
Before continuing, it is worth highlighting that Assumption~\ref{assump:HUE_interference} is introduced solely for analytical convenience. In particular, the scaling $\rho_{ku}=\mathcal O(1/U)$ is not intended to represent a physical requirement on the HUEs-to-BUEs users interference links. Rather, it should be viewed as a technical device that ensures that the aggregate interference remains well behaved as $U$ grows, thereby enabling the application of the law of large numbers in the subsequent asymptotic analysis. The assumption therefore serves as a mathematical tool for deriving the large-system characterization, rather than as a constraint imposed by the underlying propagation model.
\end{remark}

Since the HAPS--HUE channels are deterministic, the only random term in
$\mathrm{SINR}_k^{(\rm HUE,DL)}$ is the aggregate uplink interference from the BUEs.
Under Assumption~\ref{assump:HUE_interference} and assuming that $\{z_{ku}\}_{u=1}^{U}$ are independent, the following holds as $U$ grows to infinity:
\begin{equation}
\sum_{u=1}^{U}
p_u|\hat h_{ku}|^2-\sum_{u=1}^U
p_u\rho_{ku}\xrightarrow[U\rightarrow\infty]{a.s.}
0.
\end{equation}
Therefore, assuming further that the spectral norm of  $\frac{1}{\sqrt{K}}\hat{\bf H}^{({\rm HAPS})}$ remains bounded as $K$ and $M$ increase, we have that
\begin{equation}
\rm{SINR}_k^{(\rm HUE,DL)}-\overline{\rm{SINR}}_k^{(\rm HUE,DL)} \xrightarrow[U\rightarrow\infty]{a.s.} 0
\end{equation}
where
\begin{equation}\
\begin{aligned}
&\overline{\rm{SINR}}_k^{(\rm HUE,DL)}=\\
&\frac{\frac{P^{(\rm HAPS)}}{K}|(\hat{\mathbf{h}}_k^{\rm HAPS)})^H\mathbf{w}_k^{(\rm HAPS)}|^2}{\sum_{i\neq k}\frac{P^{(\rm HAPS)}}{K}|(\hat{\mathbf{h}}_k^{(\rm HAPS)})^H\mathbf{w}_i^{(\rm HAPS)}|^2+\sum_{u\in \mathcal{U}}p_u\rho_{ku}+1}
\end{aligned}
\end{equation}

\subsubsection{BS Uplink}
The received signal at the BS is expressed as
\begin{equation}
\begin{aligned}
\mathbf{y}^{(\rm BS)} 
=& \sum_{u\in \mathcal{U}} \sqrt{p}_u\,\mathbf{h}_u^{(\rm BS)} x_u \\
&+ \sum_{i\in \mathcal{K}} \sqrt{\frac{P^{(\rm HAPS)}}{K}}\, \mathbf{H}_{{\rm HB}}^H \mathbf{w}_i^{(\rm HAPS)} x_i^{(\rm HAPS)} 
+ \mathbf{n}^{(\rm BS)} ,
\end{aligned}
\end{equation}
where $\mathbf{h}_u^{(\rm BS)}\in \mathbb{C}^{N\times1}$ is the channel vector from BUE $u$ to BS, $\mathbf{H}_{{\rm HB}}\in \mathbb{C}^{M\times N}$ is the channel matrix from the HAPS to BS, and
$\mathbf{n}^{(\rm BS)}\sim \mathcal{CN}(0,\sigma_{BS}^2\mathbf{I}_{N})$ represent the additive white Gaussian noise at BS.
\par
Since the HAPS precoder lies in the null space of the HAPS--BS channel, the cross-tier interference from the HAPS is completely eliminated. Consequently, the BS only needs to detect the uplink signals transmitted by the BUEs. To this end, the BS employs the linear MMSE receiver. Accordingly, the uplink SINR of BUE $u$ is given by
\begin{equation}
\begin{aligned}
{\rm SINR}_u^{(\rm BS,UL)}=&\frac{p_u|\mathbf{g}_u^H\tilde{\mathbf{h}}_u^{(\rm BS)}|^2}{\sum_{\substack{j\in \mathcal{U} \\ j\neq u}}p_j|\mathbf{g}_u^H\tilde{\mathbf{h}}_j^{(\rm BS)}|^2+\|\mathbf{g}_u\|^2} \\
=&p_u(\tilde{\mathbf{h}}_u^{(\rm BS)})^H\mathbf\Psi_u\tilde{\mathbf{h}}_u^{(\rm BS)}.
\label{eq:SINR}
\end{aligned}
\end{equation}
where $\tilde{\mathbf{h}}_u^{(\rm BS)}=\sqrt{\tilde{\rho}_u^{(\rm BS)}/\beta_u^{(\rm BS)}}\mathbf{h}_u^{(\rm BS)}$ with $\tilde{\rho}_u^{(\rm BS)}=\frac{\beta_u^{(\rm BS)}}{\sigma^2_{BS}}$, 
the MMSE combining vector is
$$\mathbf{g}_u=[\sum_{j\in \mathcal{U}} p_j\tilde{\mathbf{h}}_j^{(\rm BS)}(\tilde{\mathbf{h}}_j^{(\rm BS)})^H+  \mathbf{I}_N]^{-1}\tilde{\mathbf{h}}_u^{(\rm BS)}$$ 
and 
$$
\mathbf\Psi_u=[\sum_{j\neq u}p_j\tilde{\mathbf{h}}_j^{(\rm BS)}(\tilde{\mathbf{h}}_j^{(\rm BS)})^H+  \mathbf{I}_N]^{-1}
$$
Unlike the deterministic HAPS links, the BS channels experience spatially correlated Rician fading, making the uplink SINR a random quantity. The normalized Kronecker BS channel is modeled as
\begin{equation}
\tilde{\mathbf h}_{l}^{(\rm BS)}
=
\sqrt{\tilde{\rho}_{l}^{(\rm BS)}}
\left(
\sqrt{\frac{1}{1+\kappa_l}}
\mathbf h_l^{\rm NLoS}
+
\sqrt{\frac{\kappa_l}{1+\kappa_l}}
\mathbf h_l^{\rm LoS}
\right),
\end{equation}
where $\mathbf{h}_l^{\text{NLoS}}=\mathbf{R}_l^{\frac{1}{2}}\mathbf{z}_l$.
For each BUE $i=1,\cdots,U$, define
\begin{equation}
\boldsymbol{\Omega}_{i}
=
\frac{
p_{i}\tilde{\rho}_{i}^{(\rm BS)}
}{
1+\kappa_{i}
}
\mathbf R_{i},
\end{equation}
and
\begin{equation}
\mathbf a_i
=
\sqrt{
\frac{
p_i\tilde{\rho}_i^{(\rm BS)}
\kappa_i
}{
1+\kappa_i
}
}
\mathbf h_i^{\rm LoS}.
\end{equation}
Collecting the deterministic LoS components, define
\begin{equation}
\mathbf A
=
\left[
\mathbf a_1,\ldots,\mathbf a_U
\right].
\end{equation}
To facilitate the asymptotic performance analysis, we adopt the following assumption, which is introduced for technical reasons, as discussed in Remark~\ref{remark:1}
\begin{assumption}
\label{ass:regime}
There exists $\tilde{\rho}_{\rm min}$ and $\tilde{\rho}_{\rm max}$ such that:
\begin{equation}
\tilde{\rho}_{\rm min}\leq \min_{u=1,\cdots,U}\tilde{\rho}_u^{(\rm BS)}\leq \max_{u=1,\cdots,U}\tilde{\rho}_u^{(\rm BS)}<\tilde{\rho}_{\rm max}
\end{equation}
Furthermore, the transmit powers are uniformly bounded, i.e.,
\begin{equation}
\sup_{u=1,\cdots,U} p_u < \infty.
\end{equation}
Finally, we assume that as $M$ and $U$ grow large:
$$
\limsup_{M,U}\|\frac{\bf A}{\sqrt{U}}\|<\infty
$$
and 
$$
0<\lim\inf\min_{u=1,\cdots,U}\frac{1}{M}{\rm tr}\ \boldsymbol{\Omega}_u\leq \max_{u=1,\cdots,U}\|\boldsymbol{\Omega}_u\|<\infty
$$
\end{assumption}
As seen in \eqref{eq:SINR}, the SINR of BUE $u$ is expressed as as a standard quadratic form involving the matrix $\boldsymbol{\Psi}_u$, which has a resolvent-like structure extensively studied in the random matrix theory literature \cite{liu2026asymptotic,Kammoun2019RZF,kammoun2020asymptotic,sanguinetti2018theoretical,couillet2011random}.   In the large-system regime, this quadratic form converges to a deterministic quantity that depends only on the large-scale channel statistics. To pave the way for the derivation of this deterministic equivalent, we introduce the following deterministic quantities, defined as the solutions to a system of fixed-point equations. More specifically, for $z\in \mathbb{R}^{-}$, define $\{\delta_i(z)\}_{i=1}^U$ and $\{\tilde{\delta}_i(z)\}_{i=1}^{U}$ the solutions to the following system of equations:
\begin{equation}
\begin{cases}
\delta_{i}(z)
=
\operatorname{tr}
\left(
\boldsymbol{\Omega}_{i}
\boldsymbol{\Theta}(z)
\right),
&
i=1,\ldots,U,
\\[1mm]
\widetilde{\delta}_{i}(z)
=
[\widetilde{\boldsymbol{\Theta}}(z)]_{i,i},
&
i=1,\ldots,U.
\end{cases}
\end{equation}
where $\boldsymbol{\Theta}(z)$ and $\widetilde{\boldsymbol{\Theta}}(z)$ are given by:
\begin{equation}
\boldsymbol{\Theta}(z)
=
\left(
\mathbf F^{-1}(z)
-
z\mathbf A
\widetilde{\mathbf F}(z)
\mathbf A^H
\right)^{-1},
\end{equation}

\begin{equation}
\widetilde{\boldsymbol{\Theta}}(z)
=
\left(
\widetilde{\mathbf F}^{-1}(z)
-
z\mathbf A^H
\mathbf F(z)
\mathbf A
\right)^{-1}.
\end{equation}
with
\begin{equation}
\mathbf F(z)
=
\left[
-z
\left(
\mathbf I_N
+
\sum_{i=1}^{U-1}
\boldsymbol{\Omega}_{i}
\widetilde{\delta}_{i}(z)
\right)
\right]^{-1},
\end{equation}

\begin{equation}
\widetilde{\mathbf F}(z)
=
\operatorname{diag}
\left(
\frac{-1}
{z\left(1+\delta_{i}(z)\right)};
\,1\leq i\leq U
\right),
\end{equation}
Then, by combining Lemma~\ref{lem:stieltjes_LoS} and Proposition~\ref{prop:LOO_quadratic_form}, we obtain under the regime specified at Assumption \ref{ass:regime},
\begin{equation}
\mathrm{SINR}_u^{(\rm BS,UL)}
-
\overline{\mathrm{SINR}}_u^{(\rm BS,UL)}
\xrightarrow[N,U\to\infty]{a.s.}
0,
\end{equation}
where
\begin{equation}
\label{eq:DE_UL_BUE}
\overline{\mathrm{SINR}}_u^{(\rm BS,UL)}
=
\delta_u(-1)
+
\frac{
\left(1+\delta_u(-1)\right)q_u(-1)
}{
1+\delta_u(-1)-q_u(-1)
}.
\end{equation}
\begin{equation}
   q_u= \mathbf a_u^H
\boldsymbol{\Theta}(-1)
\mathbf a_u
\end{equation}

 \subsection{Analysis for Transmission Interval $T_2$}
During the second transmission interval $T_2$, the HAPS operates in the uplink and receives signals from the HUEs, while the terrestrial BS performs downlink transmission to the BUEs. Since the HAPS and BS employ different transmission and reception strategies, the performance of the two links is analyzed separately in the following.
\subsubsection{HAPS Uplink}
Due to the downward-oriented antenna pattern of the terrestrial BS and the large propagation distance between the BS and the HAPS, the interference from the BS downlink transmission to the HAPS is negligible. Therefore, the received signal at the HAPS is given by
\begin{equation}
\mathbf{y}^{(\rm HAPS)}
=\sum_{k\in \mathcal{K}}\sqrt{p_k}\mathbf{h}_k^{(\rm HAPS)}x_k
+\mathbf{n}^{(\rm HAPS)},
\end{equation}
where $\mathbf h_k^{(\rm HAPS)}$ denotes the channel between HUE $k$ and the HAPS, $p_k$ is the transmit power of HUE $k$, $x_k\sim\mathcal{CN}(0,1)$ is the transmitted symbol, and
$\mathbf n^{(\rm HAPS)}
\sim \mathcal{CN}(\mathbf0,\sigma_{\rm HAPS}^2\mathbf I_M)$
is the additive white Gaussian noise vector.
\par
The HAPS employs the linear MMSE receiver for uplink signal detection. Let
$\mathbf V = [\mathbf v_1,\ldots,\mathbf v_K]$ denote the receive combining matrix. 
The uplink SINR of the $k$-th HUE is then given by
\begin{equation}
\begin{aligned}
&{\rm SINR}_k^{(\rm HAPS,UL)}=\frac{p_k|\mathbf{v}_k^H\bar{\mathbf{h}}_k^{(\rm HAPS)}|^2}{\sum_{\substack{i\in \mathcal{K}\\ i\neq k}}p_i|\mathbf{v}_k^H\bar{\mathbf{h}}_i^{(\rm HAPS)}|^2+\|\mathbf{v}_k\|^2} \\
& = p_k
(\bar{\mathbf h}_k^{(\mathrm{HAPS})})^H(
\sum_{i\neq k}p_i\bar{\mathbf h}_i^{(\mathrm{HAPS})}(\bar{\mathbf h}_i^{(\mathrm{HAPS})})^H
+\mathbf I_M
)^{-1}
\bar{\mathbf h}_k^{(\mathrm{HAPS})},
\end{aligned}
\end{equation}
where
$\bar{\mathbf h}_k^{(\rm HAPS)}
=\sqrt{\bar{\rho}_k^{(\rm HAPS)}}\mathbf h_{k,\rm HAPS}^{\rm LoS}$,
with
$\bar{\rho}_k^{(\rm HAPS)}
=\beta_k^{(\rm HAPS)}/\sigma_{\rm HAPS}^2$.
The MMSE receive vector for HUE $k$ is given by
\begin{equation}
\mathbf v_k=
\left(
\sum_{i\in \mathcal{K}}
p_i
\bar{\mathbf h}_i^{(\rm HAPS)}
(\bar{\mathbf h}_i^{(\rm HAPS)})^H
+\mathbf I_M
\right)^{-1}
\bar{\mathbf h}_k^{(\rm HAPS)}.
\end{equation}
Since the HAPS channels consist only of deterministic LoS components, the uplink SINRs are deterministic for a given user deployment and can be computed directly from the above expressions without invoking random matrix theory.

\subsubsection{BS Downlink}
The received signal at BUE $u$
\begin{equation}
\begin{aligned}
y_u^{(\rm BUE)} 
=& \sum_{j\in \mathcal{U}} 
\sqrt{\frac{P^{(\rm BS)}}{U}} (\mathbf{h}_u^{(\rm BS)})^H \mathbf{w}_j^{(\rm BS)} x_j^{(\rm BS)}\\
 &+ \sum_{k\in \mathcal{K}} \sqrt{p_k} h_{ku} x_k
+ n_u ,
\end{aligned}
\end{equation}
where $\mathbf h_u^{(\rm BS)}$ denotes the channel vector between the BS and the $u$-th BUE, $\mathbf w_j^{(\rm BS)}$ is the precoding vector intended for BUE $j$, $P^{(\rm BS)}$ is the total BS transmit power, $x_j^{(\rm BS)}\sim\mathcal{CN}(0,1)$ is the transmitted symbol, $h_{ku}$ denotes the interference channel from HUE $k$ to BUE $u$, and $n_u\sim\mathcal{CN}(0,\sigma^2)$ is the additive white Gaussian noise.
\par
Let
$
\mathbf W^{(\rm BS)}
=
\left[
\mathbf w_1^{(\rm BS)},
\ldots,
\mathbf w_U^{(\rm BS)}
\right]
$ denote the ZF precoding matrix. The downlink SINR of BUE $u$ is:
\begin{equation}
\label{eq:BUE_DL_SINR}
\begin{aligned}
&\rm{SINR}_u^{(\rm{BUE},DL)}=\\
&\frac{\frac{P^{(\rm{BS})}}{U}|(\hat{\mathbf{h}}_u^{(\rm BS)})^H\mathbf{w}_u^{(\rm BS)}|^2}{\sum_{\substack{j\in \mathcal{U}\\ j\neq u}}\frac{P^{(\rm BS)}}{U}|(\hat{\mathbf{h}}_u^{(\rm BS)})^H\mathbf{w}_j^{(\rm BS)}|^2+\sum_{k \in \mathcal{K}} p_k|\hat{h}_{ku}|^2+1}
\end{aligned}
\end{equation}
where $\hat{\mathbf{h}}_u^{(\rm BS)}= \sqrt{\rho_u^{(\rm BS)}/\beta_u^{(\rm BS)}}\mathbf{h}_u^{(\rm BS)}$ with $\rho_u^{(\rm BS)}=\frac{\beta_u^{(\rm BS)}}{\sigma^2}$, and $\hat{h}_{ku}=\sqrt{\rho_{ku}}z_{ku}$, with $\rho_{ku}=\frac{\beta_{ku}}{\sigma^2}$.
 Define the normalized BS--BUE channel matrix as
\begin{equation}
\hat{\mathbf H}_{\rm BUE}^{(\rm BS)}
=
\left[
\hat{\mathbf h}_1^{(\rm BS)},
\ldots,
\hat{\mathbf h}_U^{(\rm BS)}
\right]
\in\mathbb C^{N\times U}.
\end{equation}
The unnormalized ZF precoding matrix is given by
\begin{equation}
\widetilde{\mathbf W}^{(\rm BS)}
=
\hat{\mathbf H}_{\rm BUE}^{(\rm BS)}
\left[
(\hat{\mathbf H}_{\rm BUE}^{(\rm BS)})^H
\hat{\mathbf H}_{\rm BUE}^{(\rm BS)}
\right]^{-1},
\end{equation}
where
$\widetilde{\mathbf W}^{(\rm BS)}
=
[\widetilde{\mathbf w}_1^{(\rm BS)},\ldots,
\widetilde{\mathbf w}_U^{(\rm BS)}]$.
The normalized ZF beamforming vector for BUE $u$ is then
\begin{equation}
\mathbf w_u^{(\rm BS)}
=
\frac{\widetilde{\mathbf w}_u^{(\rm BS)}}
{\|\widetilde{\mathbf w}_u^{(\rm BS)}\|},
\qquad u\in\mathcal U.
\end{equation}
Since ZF precoding eliminates the intra-cell interference among the BUEs, equation \eqref{eq:BUE_DL_SINR} reduces to
\begin{equation}
\mathrm{SINR}_u^{(\rm BUE,DL)}
=
\frac{
\frac{P^{(\rm BS)}}{U}
\left|
(\hat{\mathbf h}_u^{(\rm BS)})^H
\mathbf w_u^{(\rm BS)}
\right|^2
}{
\sum_{k\in \mathcal{K}}
p_k|\hat h_{ku}|^2+1
}.
\end{equation}

To facilitate the large-system analysis, we employ the analytical Kronecker channel model introduced in Section~II, which preserves the spatial correlation characteristics of the 3GPP channel while providing a tractable analytical representation. 
In addition, we define that
\begin{equation}
\hat{\boldsymbol{\Omega}}_u
=
\frac{\rho_u^{(\rm BS)}\mathbf R_u}{1+\kappa_u},
\qquad
\hat{\mathbf{a}}_u
=
\frac{1}{\sqrt U}
\sqrt{
\frac{\rho_u^{(\rm BS)}\kappa_u}{1+\kappa_u}
}
\mathbf h_u^{\rm LoS}.
\end{equation}
Collecting the deterministic LoS components of all BUEs, define
\begin{equation}
\hat{\mathbf A}
=
\left[
\hat{\mathbf a}_1,\ldots,\hat{\mathbf a}_U
\right].
\end{equation}

\par
The ZF precoder requires the inversion of the
channel Gram matrix. To ensure that this inversion is
well-defined and to characterize its asymptotic behavior,
we impose the following assumption on the system dimensions
and the user channel covariance matrices. In particular,
the assumption guarantees that the number of transmit
antennas exceeds the number of users and that each
covariance matrix has a sufficiently large number of
eigenvalues bounded away from zero.

\begin{assumption}
\label{ass:zeroforcing}
Assume that $\lim\inf\frac{N}{U}>1$, and that there exists $\varphi$ and $\eta$ such that:
$$
\min_{i=1,\cdots,U}\frac{1}{U}{\rm tr} \hat{\boldsymbol{\Omega}}_i(\hat{\boldsymbol{\Omega}}_i+\varphi{\bf I}_N)^{-1}\geq 1+\eta
$$
\end{assumption}
\begin{theorem}
\label{thm:BUE_DL_SINR}
Under Assumptions~\ref{assump:HUE_interference} and~\ref{ass:zeroforcing},
and assuming that $\rho_u^{(\rm BS)}$, $\hat{\boldsymbol{\Omega}}_u$, and
$\hat{\mathbf A}$ satisfy the conditions in Assumption~\ref{ass:regime}, 
the downlink SINR of BUE $u\in\mathcal U$ satisfies
\begin{equation}
\mathrm{SINR}_u^{(\rm BUE,DL)}
-
\overline{\mathrm{SINR}}_u^{(\rm BUE,DL)}
\xrightarrow[N,U,K\to\infty]{a.s.}
0,
\end{equation}
where
\begin{equation}
\label{eq:DE_SINR_BUE}
\overline{\mathrm{SINR}}_u^{(\rm BUE,DL)}
=
\frac{
P^{(\rm BS)}
\stackrel{\circ}{\mu}_u
}{
\displaystyle
\sum_{k\in\mathcal K}
p_k\rho_{ku}
+1
},
\end{equation}
with
\begin{equation}
\stackrel{\circ}{\mu}_u
=
\underline{\delta}_u
+
\frac{
\underline{\delta}_u
\underline{q}_u
}{
\underline{\delta}_u
-
\underline{q}_u
}.
\end{equation}
Here,
\begin{equation}
\underline{\delta}_u
=
\frac{1}{U}
\operatorname{tr}
\left(
\hat{\boldsymbol{\Omega}}_u
\underline{\boldsymbol{\Theta}}
\right),
\qquad
\underline q_u
=
\hat{\mathbf{a}}_u^H
\underline{\boldsymbol{\Theta}}
\hat{\mathbf{a}}_u,
\end{equation}
The matrix $\underline{\boldsymbol{\Theta}}$ is the full-system ZF
deterministic equivalent characterized by the fixed-point equations as follows:
\begin{equation}
\label{eq:fixed_point_Theta}
\begin{cases}
\displaystyle
\underline{\delta}_{i}
=
\frac{1}{U}
\operatorname{tr}
\left(
\hat{\boldsymbol{\Omega}}_{i}
\underline{\boldsymbol{\Theta}}
\right),
&
i=1,\ldots,U,
\\[2mm]
\displaystyle
\underline{\widetilde{\delta}}_{i}
=
\left[
\underline{\widetilde{\boldsymbol{\Theta}}}
\right]_{i,i},
&
i=1,\ldots,U.
\end{cases}
\end{equation}
where
\begin{equation}
\underline{\mathbf F}
=
\left(
\mathbf I_N
+
\frac{1}{U}
\sum_{i=1}^{U}
\hat{\boldsymbol{\Omega}}_{i}
\underline{\widetilde{\delta}}_{i}
\right)^{-1},
\end{equation}

\begin{equation}
\underline{\widetilde{\mathbf F}}
=
\operatorname{diag}
\left(
\frac{1}{\underline{\delta}_{i}};
\, i=1,\ldots,U
\right),
\end{equation}

\begin{equation}
\underline{\boldsymbol{\Theta}}
=
\left(
\underline{\mathbf F}^{-1}
+
\hat{\mathbf A}
\underline{\widetilde{\mathbf F}}
\hat{\mathbf A}^H
\right)^{-1},
\end{equation}

and
\begin{equation}
\underline{\widetilde{\boldsymbol{\Theta}}}
=
\left(
\underline{\widetilde{\mathbf F}}^{-1}
+
\hat{\mathbf A}^H
\underline{\mathbf F}
\hat{\mathbf A}
\right)^{-1}.
\end{equation}

\end{theorem}
\begin{proof}
See Appendix~\ref{app:proof_of_BUE_DL_SINR}.
\end{proof}

\section{TDD Scheme}
\label{section:TDD}
In this section, we consider conventional TDD solely as a benchmark for comparison with the proposed reverse-TDD scheme. Under conventional TDD, the HAPS and the terrestrial BS operate in the same transmission direction over the same time-frequency resources. Consequently, both tiers transmit simultaneously in the DL and receive simultaneously in the UL, resulting in cross-tier interference between the HAPS and terrestrial networks. We do not develop a separate theoretical analysis for the conventional TDD configuration; instead, its performance is evaluated numerically in the next section and used as a reference for assessing the benefits of the proposed reverse-TDD operation.
\subsection{Downlink Transmission}
During the DL phase, the HAPS serves the HUEs while the terrestrial BS simultaneously serves the BUEs. The received signal at HUE $k\in\mathcal K$ is given by
\begin{equation}
\begin{aligned}
y_k^{(\rm HUE, TDD)} 
=& \sum_{i\in \mathcal{K}}
\sqrt{\frac{P^{(\rm HAPS)}}{K}} (\mathbf{h}_k^{(\rm HAPS)})^H \mathbf{w}_i^{(\rm HAPS,TDD)} x_i^{(\rm HAPS)}\\
+& \sum_{u\in \mathcal{U}}\sqrt{\frac{P^{(\rm BS)}}{U}} (\mathbf{h}_k^{(\rm BS)})^H \mathbf{w}_u^{(\rm BS,TDD)} x_u^{(\rm BS)} 
+ n_k
\end{aligned}
\end{equation}
where $\mathbf{h}_k^{(\rm HAPS)}\in\mathbb{C}^{M\times1}$ and
$\mathbf{h}_k^{(\rm BS)}\in\mathbb{C}^{N\times1}$ denote the channel vectors
from the HAPS and the terrestrial BS to HUE $k$, respectively.
$\mathbf{w}_i^{(\rm HAPS,TDD)}\in\mathbb{C}^{M\times1}$ and
$\mathbf{w}_u^{(\rm BS,TDD)}\in\mathbb{C}^{N\times1}$ denote the ZF precoding
vectors for HUE $i$ and BUE $u$, respectively.
The transmitted symbols satisfy
$x_i^{(\rm HAPS)}\sim\mathcal{CN}(0,1)$ and
$x_u^{(\rm BS)}\sim\mathcal{CN}(0,1)$.
The additive white Gaussian noise is modeled as
$n_k\sim\mathcal{CN}(0,\sigma^2)$.

\par
Similarly, the received signal at BUE $u$ is given by
\begin{equation}
\begin{aligned}
y_u^{(\rm BUE, TDD)} 
&= \sum_{j\in \mathcal{U}} 
\sqrt{\frac{P^{(\rm BS)}}{U}} (\mathbf{h}_u^{(\rm BS)})^H \mathbf{w}_j^{(\rm BS,TDD)} x_j^{(\rm BS)} \\
+& \sum_{i \in \mathcal{K}}
\sqrt{\frac{P^{(\rm HAPS)}}{K}} (\mathbf{h}_u^{(\rm HAPS)})^H \mathbf{w}_i^{(\rm HAPS,TDD)} x_i^{(\rm HAPS)}
+ n_u ,
\end{aligned}
\end{equation}
where $\mathbf{h}_u^{(\rm BS)}\in\mathbb{C}^{N\times1}$ and
$\mathbf{h}_u^{(\rm HAPS)}\in\mathbb{C}^{M\times1}$ denote the channel vectors
from the terrestrial BS and the HAPS to BUE $u$, respectively, and
$n_u\sim\mathcal{CN}(0,\sigma^2)$ denotes the additive white Gaussian noise.
The DL SINR of HUE $k$ is
\begin{equation}
\mathrm{SINR}_{k}^{(\mathrm{HUE},\mathrm{DL}, \mathrm{TDD})}
=
\frac{
\frac{P^{(\mathrm{HAPS})}}{K}
|(\mathbf{h}_{k}^{(\mathrm{HAPS})})^{H}
\mathbf{w}_{k}^{(\mathrm{HAPS,TDD})}|^{2}}
{I_k^{(\mathrm{HUE},\mathrm{DL})}
+\sigma^2
}.
\end{equation}

where 
\begin{equation}
\begin{aligned}
I_k^{(\mathrm{HUE},\mathrm{DL})}=&\frac{P^{(\mathrm{HAPS})}}{K}
\sum_{\substack{i\in\mathcal K\\i\neq k}}
|
(\mathbf{h}_{k}^{(\mathrm{HAPS})})^{H}
\mathbf w_{i}^{(\mathrm{HAPS,TDD})}
|^{2}\\
+&
\frac{P^{(\mathrm{BS})}}{U}\sum_{j\in\mathcal U}
|(\mathbf h_{k}^{(\mathrm{BS})})^{H}
\mathbf w_{j}^{(\mathrm{BS,TDD})}|^{2}
\end{aligned}
\end{equation}
Similarly, the DL SINR of BUE $u$ is

\begin{equation}
\mathrm{SINR}_{u}^{(\mathrm{BUE},\mathrm{DL},\mathrm{TDD})}
=\frac{
\frac{P^{(\mathrm{BS})}}{U}
|({\mathbf h}_{u}^{(\mathrm{BS})})^{H}
\mathbf w_{u}^{(\mathrm{BS,TDD})}|^{2}
}
{I_{u}^{(\mathrm{BUE},\mathrm{DL})}
+\sigma^2
}.
\end{equation}
where 
\begin{equation}
\begin{aligned}
I_{u}^{(\mathrm{BUE},\mathrm{DL})}= &\frac{P^{(\mathrm{BS})}}{U}
\sum_{\substack{j\in\mathcal U\\j\neq u}}
|({\mathbf h}_{u}^{(\mathrm{BS})})^{H}
\mathbf w_{j}^{(\mathrm{BS,TDD})}|^{2}\\
& +
\frac{P^{(\mathrm{HAPS})}}{K}
\sum_{i\in\mathcal K}|
({\mathbf h}_{u}^{(\mathrm{HAPS})})^{H}
\mathbf w_{i}^{(\mathrm{HAPS,TDD})}|^{2}
\end{aligned}
\end{equation}
\subsubsection{Conventional TDD}
For conventional TDD, the HAPS and terrestrial BS independently design
their ZF precoders using only the channel information of their associated
users. Hence, ZF suppresses intra-tier multi-user interference but does not
explicitly mitigate cross-tier interference.
Define the BS--BUE channel matrix as:
\begin{equation}
\mathbf H_{\rm BUE}^{(\rm BS)}
= \left[ \mathbf h_1^{(\rm BS)}, \ldots, \mathbf h_U^{(\rm BS)}
\right].
\end{equation}
The corresponding unnormalized ZF precoding matrix at the BS is
\begin{equation}
\widetilde{\mathbf{W}}^{(\mathrm{BS},\mathrm{TDD})}
={\mathbf{H}}_{\mathrm{BUE}}^{(\mathrm{BS})}
\left[
({\mathbf{H}}_{\mathrm{BUE}}^{(\mathrm{BS})})^H
{\mathbf{H}}_{\mathrm{BUE}}^{(\mathrm{BS})}
\right]^{-1},
\end{equation}
Let
$\widetilde{\mathbf w}_u^{(\rm BS,\mathrm{TDD})}$
denote its $u$-th column. The normalized precoding vector is
\begin{equation}
\mathbf{w}_u^{(\mathrm{BS},\mathrm{TDD})}
=
\frac{\widetilde{\mathbf{w}}_u^{(\mathrm{BS},\mathrm{TDD})}}
{\left\|
\widetilde{\mathbf{w}}_u^{(\mathrm{BS},\mathrm{TDD})}
\right\|}.
\end{equation}

Similarly, define the HAPS--HUE channel matrix as
\begin{equation}
\mathbf H_{\rm HUE}^{(\rm HAPS)}
=\left[ \mathbf h_1^{(\rm HAPS)}, \ldots, \mathbf h_K^{(\rm HAPS)}
\right].
\end{equation}
The unnormalized ZF precoding matrix at the HAPS is
\begin{equation}
\widetilde{\mathbf{W}}^{(\mathrm{HAPS},\mathrm{TDD})}
={\mathbf{H}}_{\mathrm{HUE}}^{(\mathrm{HAPS})}
\left[
({\mathbf{H}}_{\mathrm{HUE}}^{(\mathrm{HAPS})})^H
{\mathbf{H}}_{\mathrm{HUE}}^{(\mathrm{HAPS})}
\right]^{-1},
\end{equation}
and the normalized beamforming vector for HUE $k$ is
\begin{equation}
\mathbf{w}_k^{(\mathrm{HAPS},\mathrm{TDD})}
=
\frac{\widetilde{\mathbf{w}}_k^{(\mathrm{HAPS},\mathrm{TDD})}}
{\left\|
\widetilde{\mathbf{w}}_k^{(\mathrm{HAPS},\mathrm{TDD})}
\right\|}.
\end{equation}
\subsubsection{Null-Space Aided TDD}
As an interference-mitigation benchmark, we further consider a
null-space-aided TDD (NSA-TDD) scheme. Unlike conventional TDD, the HAPS sacrifices part of its available spatial degrees of freedom to project its transmission onto the null space of the HAPS--BUE channel, thereby suppressing the cross-tier interference toward the BUEs.
Define the HAPS--BUE channel matrix as
\begin{equation}
{\mathbf{H}}_{\rm BUE}^{(\rm HAPS)}
=\left[{\mathbf{h}}_1^{(\rm HAPS)},
\ldots, {\mathbf{h}}_U^{(\rm HAPS)} \right] \in\mathbb{C}^{M\times U}.
\end{equation}
The projection matrix onto the null space of
$\mathbf H_{\rm BUE}^{(\rm HAPS)}$ is given by
\begin{equation}
\begin{aligned}
\mathbf{T}^{(\mathrm{NSA})}
=&
\mathbf{I}_M
-{\mathbf{H}}_{\rm BUE}^{(\rm HAPS)}
\left[
({\mathbf{H}}_{\rm BUE}^{(\rm HAPS)})^H
{\mathbf{H}}_{\rm BUE}^{(\rm HAPS)}
\right]^{-1}\\
&\times
({\mathbf{H}}_{\rm BUE}^{(\rm HAPS)})^H.
\end{aligned}
\end{equation}
The projected HAPS--HUE channel is then
\begin{equation}
\overline{\mathbf H}_{\rm HUE}^{(\rm HAPS)}
=\mathbf{T}^{(\mathrm{NSA})}
\mathbf H_{\rm HUE}^{(\rm HAPS)}.
\end{equation}
Based on this projected channel, the HAPS employs ZF precoding as
\begin{equation}
\widetilde{\mathbf W}^{(\rm HAPS,\rm NSA)}
=
\overline{\mathbf H}_{\rm HUE}^{(\rm HAPS)}
\left[
(\overline{\mathbf H}_{\rm HUE}^{(\rm HAPS)})^H
\overline{\mathbf H}_{\rm HUE}^{(\rm HAPS)}
\right]^{-1}.
\end{equation}
The normalized precoding vector for HUE $k$ is
\begin{equation}
\mathbf w_{k}^{(\rm HAPS,\rm NSA)}
=
\frac{
\widetilde{\mathbf w}_k^{(\rm HAPS,\rm NSA)}
}{
\|
\widetilde{\mathbf w}_k^{(\rm HAPS,\rm NSA)}
\|
}.
\end{equation}
The BS employs the same ZF precoder as in conventional TDD, while the
HAPS precoder is replaced by
$\mathbf W^{(\rm HAPS,NSA)}$.
Accordingly, the HAPS-induced interference to the BUEs is eliminated.

\subsection{Uplink Transmission}
During the UL phase, the HUEs and BUEs simultaneously transmit to the HAPS and terrestrial BS, respectively. As a result, each receiver experiences both intra-tier multi-user interference and cross-tier interference. The received signal at the HAPS is
\begin{equation}
\begin{aligned}
\mathbf{y}^{(\mathrm{HAPS},\mathrm{TDD})}=&
\sum_{k\in \mathcal{K}} \sqrt{p_k}\,\mathbf{h}_k^{(\mathrm{HAPS})} x_k\\
+&\sum_{u\in \mathcal{U}}\sqrt{p_u}\,\mathbf{h}_u^{(\mathrm{HAPS})} x_u+ \mathbf{n}^{(\mathrm{HAPS})},
\end{aligned}
\end{equation}
where $\mathbf{n}^{(\mathrm{HAPS})} \sim \mathcal{CN}(\mathbf{0},\sigma_{\mathrm{HAPS}}^2\mathbf{I})$ denotes the additive noise at the HAPS.
\par
Similarly, the received signal at the BS is given by
\begin{equation}
\begin{aligned}
\mathbf{y}^{(\mathrm{BS},\mathrm{TDD})}=&
\sum_{u\in \mathcal{U}}\sqrt{p_u}\,\mathbf{h}_u^{(\mathrm{BS})} x_u\\
+&
\sum_{k\in \mathcal{K}} \sqrt{p_k}\,\mathbf{h}_k^{(\mathrm{BS})} x_k
+\mathbf{n}^{(\mathrm{BS})}
\end{aligned}
\end{equation}
where $\mathbf{n}^{(\mathrm{BS})} \sim \mathcal{CN}(\mathbf{0},\sigma_{\mathrm{BS}}^2\mathbf{I})$ is the additive noise at the BS.
\par
Let $\mathbf v_k^{(\rm TDD)}$ and
$\mathbf g_u^{(\rm TDD)}$ denote the MMSE combining vectors for HUE $k$ and BUE $u$, respectively. The corresponding uplink SINRs are given by:
\begin{equation}
\mathrm{SINR}_k^{(\mathrm{HUE},\mathrm{UL},\mathrm{TDD})}=
\frac{p_k|(\mathbf{v}^{(\rm TDD)}_k)^H{\mathbf{h}}_k^{(\mathrm{HAPS})}|^2}{I_k^{(\mathrm{HUE},\mathrm{UL})}+\sigma_{HAPS}^2\|\mathbf{v}^{(\rm TDD)}_k\|^2
}.
\end{equation}
where
\begin{equation}
\begin{aligned}
I_k^{(\mathrm{HUE},\mathrm{UL})}=&\sum_{\substack{i\in \mathcal{K}\\i\neq k}}
p_i|(\mathbf{v}^{(\rm TDD)}_k)^H{\mathbf{h}}_i^{(\mathrm{HAPS})}|^2\\
+&\sum_{u \in \mathcal{U}} p_u|(\mathbf{v}^{(\rm TDD)}_k)^H{\mathbf{h}}_u^{(\mathrm{HAPS})}|^2
\end{aligned}
\end{equation}
Similarly, the uplink SINR for BUE $u$ at the BS is given by
\begin{equation}
\mathrm{SINR}_u^{(\mathrm{BUE},\mathrm{UL},\mathrm{TDD})}
=
\frac{p_u|(\mathbf{g}^{(\rm TDD)}_u)^H {\mathbf{h}}_u^{(\mathrm{BS})}|^2}{
I_u^{(\mathrm{BUE},\mathrm{UL})}
+ \sigma_{BS}^2\|\mathbf{g}^{(\rm TDD)}_u\|^2}
.
\end{equation}
where 
\begin{equation}
I_u^{(\mathrm{BUE},\mathrm{UL})}=\sum_{\substack{j \in \mathcal{U}\\j\neq u}}
p_j|(\mathbf{g}^{(\rm TDD)}_u)^H {\mathbf{h}}_j^{(\mathrm{BS})}
|^2 + \sum_{k\in \mathcal{K}}p_k|(\mathbf{g}^{(\rm TDD)}_u)^H {\mathbf{h}}_k^{(\mathrm{BS})}|^2
\end{equation}
For the uplink reception, neither conventional TDD nor NSA-TDD performs cross-tier receive interference suppression. Accordingly, each receiver constructs its MMSE combining vector using only the CSI of its associated users.
The MMSE
combining vector for HUE $k$ at the HAPS is
\begin{equation}
\begin{aligned}
\mathbf{v}^{(\rm TDD)}_k
=&(
\sum_{i\in \mathcal{K}}p_i{\mathbf{h}}_i^{(\mathrm{HAPS})}({\mathbf{h}}_i^{(\mathrm{HAPS})})^H
+ \sigma^2_{HAPS}\mathbf{I})^{-1}
{\mathbf{h}}_k^{(\mathrm{HAPS})}.
\end{aligned}
\end{equation}
Likewise, the MMSE combining vector for BUE $u$ at the terrestrial BS is
\begin{equation}
\begin{aligned}
\mathbf{g}^{(\rm TDD)}_u
=&(
\sum_{j\in \mathcal{U}}p_j{\mathbf{h}}_j^{(\mathrm{BS})}({\mathbf{h}}_j^{(\mathrm{BS})})^H
+\sigma^2_{BS}\mathbf{I})^{-1}
{\mathbf{h}}_u^{(\mathrm{BS})}.
\end{aligned}
\end{equation}

\section{Numerical Results}
\label{section:Simulation Result}
We consider an integrated HAPS--terrestrial network consisting of one HAPS equipped with $M$ antennas and one terrestrial BS equipped with $N$ antennas. The HAPS covers a circular area with radius $R_{\rm HAPS}$, while the BS serves a local coverage area with radius $R_{\rm BS}$. The $U$ BUEs are uniformly distributed within the BS coverage area, whereas the $K$ HUEs are uniformly distributed within the HAPS coverage area but outside the BS coverage region. 
Unless otherwise specified, the simulation parameters are adopted from the 3GPP specifications in \cite{3gpp_tr38863,3gpp_tr38811} and are summarized in Table~\ref{tab:simulation_parameters}.
\begin{table}[!h]
\centering
\caption{Simulation Parameters}
\label{tab:simulation_parameters}
\begin{tabular}{|p{0.6\columnwidth}|c|}
\hline
\textbf{Parameter} & \textbf{Value} \\ \hline
Carrier frequency $f_c$ & 2 GHz \\ \hline
HAPS altitude & 20 km \\ \hline
BS height & 25 m \\ \hline
UE height & 1.5 m \\ \hline
HAPS coverage radius $R_{\rm HAPS}$ & 50 km \\ \hline
BS coverage radius $R_{\rm BS}$& 5000 m \\ \hline
DL bandwidth & 20 MHz \\ \hline
UL bandwidth & 1.08 MHz \\ \hline
UE transmit power & 23 dBm \\ \hline
UE antenna gain & 0 dBi \\ \hline
UE noise figure & 9 dB \\ \hline
Building penetration loss & 15 dB \\ \hline
BS transmit power & 47 dBm \\ \hline
BS noise figure & 5 dB \\ \hline
HAPS element gain $G_{E,\max}$ & 8 dBi \\ \hline
$\theta_{3\mathrm{dB}}$ & $65^\circ$ \\ \hline
$\phi_{3\mathrm{dB}}$ & $65^\circ$ \\ \hline
HAPS noise figure & 5 dB \\ \hline
Transmit power per HAPS antenna panel & 43 dBm \\ \hline
The number of HAPS antennas & 200 \\ \hline
\end{tabular}
\end{table}

To evaluate the SINR distribution among users, we consider the empirical cumulative distribution function (CDF), defined as
\begin{equation}
F_{\mathrm{SINR}}(\gamma_{\rm th})
=
\mathbb{E}
\left[
\frac{\#\{\mathrm{SINR}\le\gamma_{\rm th}\}}
{N_{\rm UE}}
\right],
\end{equation}
where $\gamma_{\rm th}$ denotes the SINR threshold, $\#\{\cdot\}$ denotes the cardinality of a set, and $N_{\rm UE}$ denotes the number of users in the corresponding tier (i.e., $K$ for HUEs and $U$ for BUEs).

\subsection{Performance Evaluation}
In addition to the conventional TDD and the proposed NSA-TDD schemes, two interference-free reference configurations are considered for comparison. In the HAPS-only configuration, the terrestrial network is completely removed, such that the system consists only of the HAPS and its associated HUEs. Hence, no terrestrial BSs or BUEs are present, and no cross-tier interference arises. Conversely, in the BS-only configuration, the HAPS network is completely removed, and only the terrestrial BS and its associated BUEs are considered. In this case, the HAPS and HUEs are absent, and the terrestrial network operates without cross-tier interference from the HAPS tier.
These two configurations are introduced to isolate the performance of each network in the absence of cross-tier interference and thereby provide reference points for quantifying the performance degradation caused by HAPS-terrestrial coexistence and the extent to which the proposed NSA-TDD scheme mitigates this degradation.

\begin{figure*}[!t]
\centering
\subfigure[CDF of downlink SINR for HUE.]{
\includegraphics[width=2.5in]{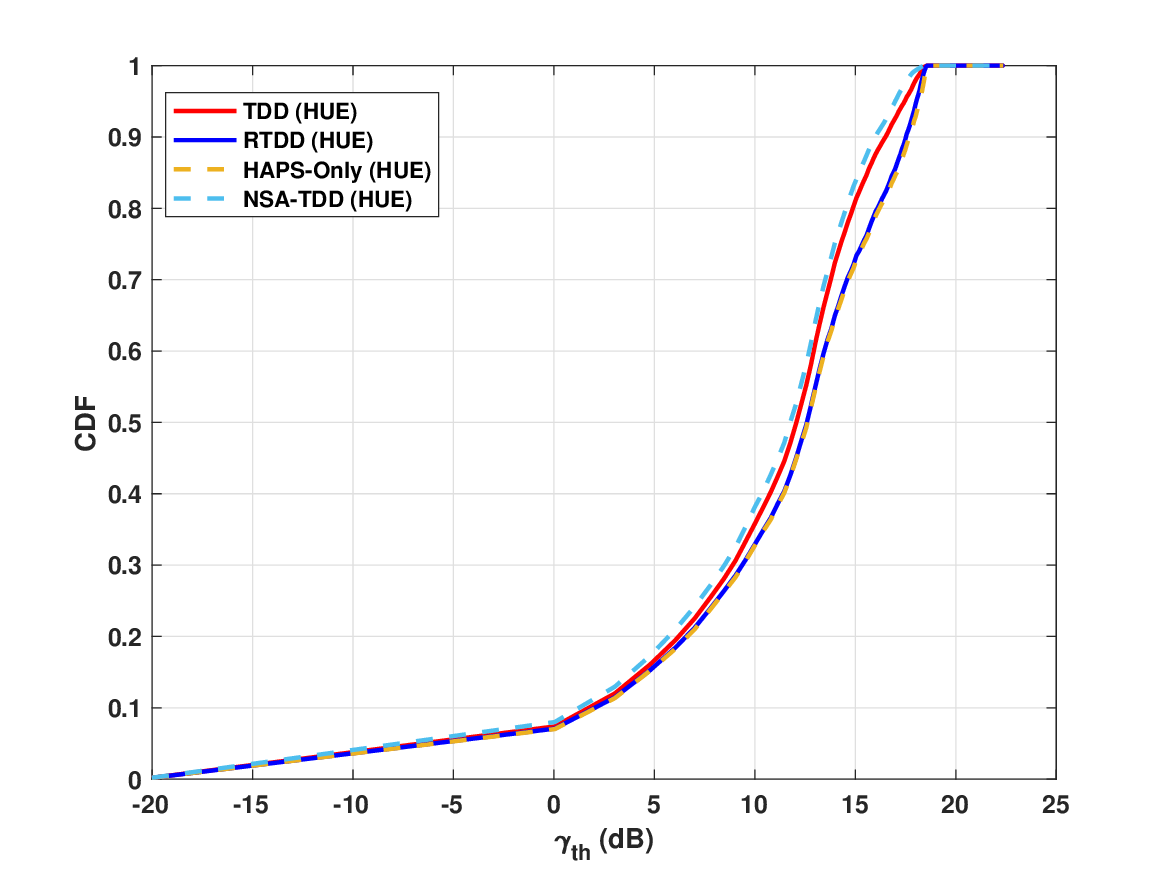}
\label{Fig:HUE_DL}
}
\hspace{-0.4cm}
\subfigure[CDF of downlink SINR for BUE.]{
\includegraphics[width=2.5in]{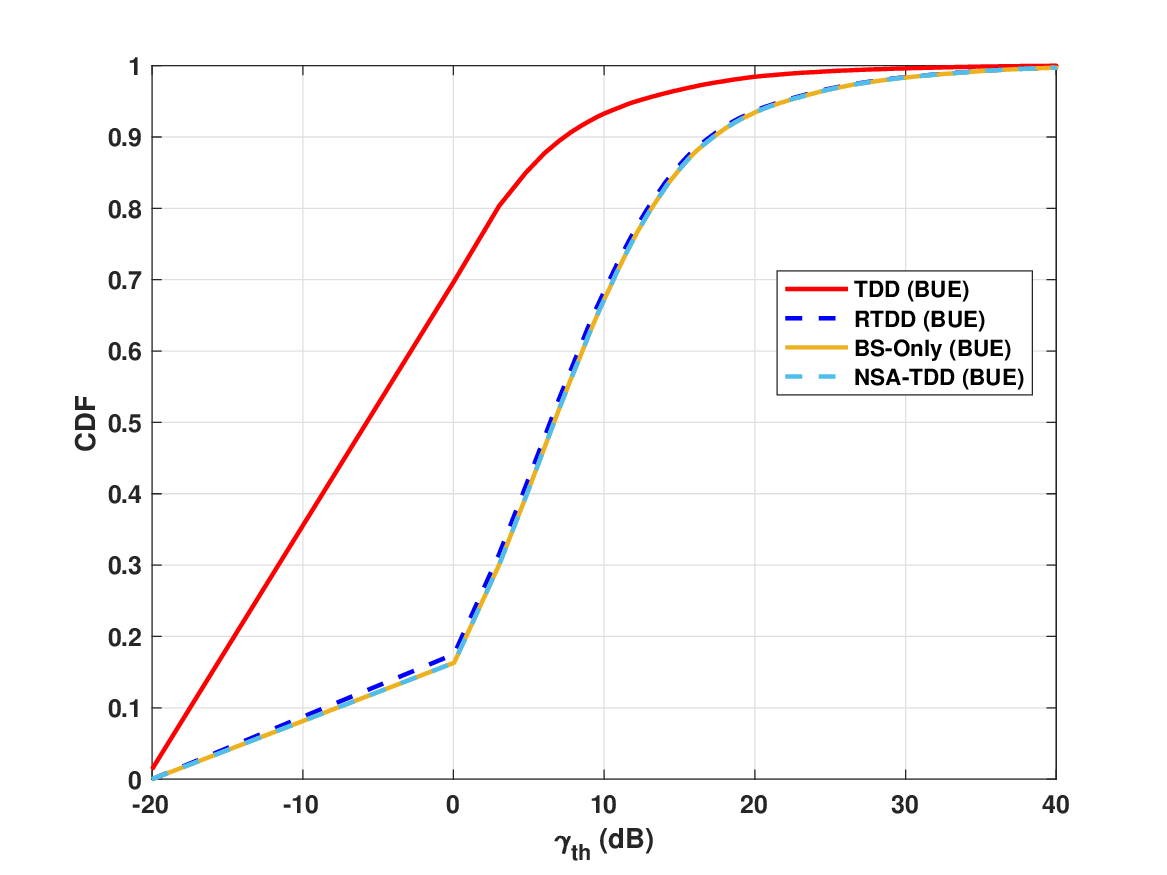}
\label{Fig:BUE_DL}
}

\subfigure[CDF of uplink SINR for HUE.]{
\includegraphics[width=2.5in]{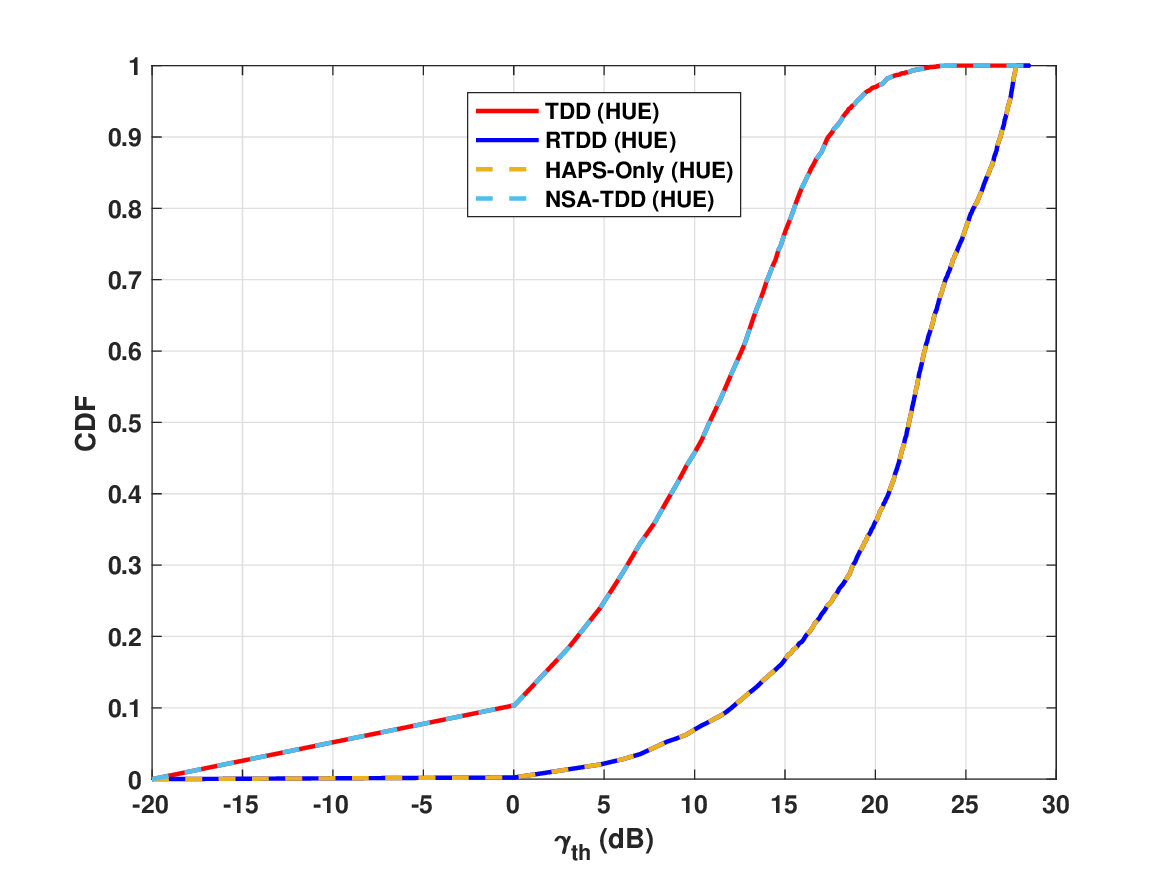}
\label{Fig:HUE_UL}
}
\hspace{-0.4cm}
\subfigure[CDF of uplink SINR for BUE.]{
\includegraphics[width=2.5in]{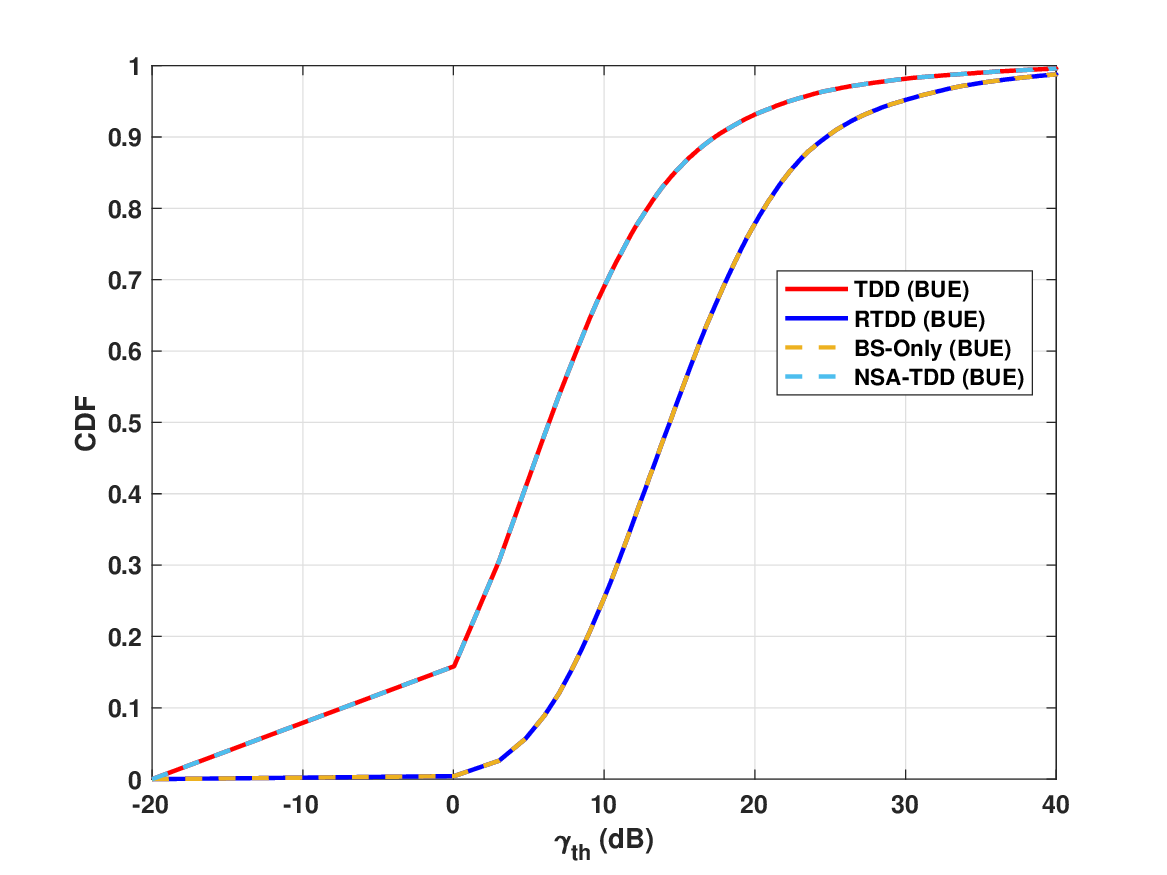}
\label{Fig:BUE_UL}
}

\caption{CDFs of the downlink and uplink SINR under different transmission schemes.}
\label{Fig:RTDD}
\end{figure*}

Fig.~\ref{Fig:HUE_DL} and Fig.~\ref{Fig:BUE_DL} show the CDFs of the downlink SINR for HUEs and BUEs under different transmission schemes. As shown in Fig.~\ref{Fig:HUE_DL}, both the conventional TDD and the proposed NSA-TDD exhibit degraded downlink SINR performance for HUEs compared with the proposed RTDD protocol. For conventional TDD, the performance degradation is caused by the cross-tier interference from the terrestrial BS. 
In contrast, although NSA-TDD suppresses the interference from the HAPS to the BUEs, the null-space projection constrains the HAPS precoder and reduces the available spatial degrees of freedom, resulting in a slight loss in beamforming gain for the HUEs. Moreover, implementing NSA-TDD requires the HAPS to acquire the CSI of the BUEs, which introduces additional channel estimation and signaling overhead.
By comparison, the proposed RTDD protocol experiences negligible cross-tier interference owing to the large propagation loss and building penetration loss between the BS and the HUEs, thereby achieving better downlink SINR performance.

A significant performance improvement is observed for BUEs, as illustrated in Fig.~\ref{Fig:BUE_DL}. Under the conventional TDD scheme, the strong LoS downlink transmission from the HAPS causes severe cross-tier interference to the BUEs, resulting in a substantial degradation in the downlink SINR. In the RTDD protocol, however, the HAPS operates in the uplink while the BS performs downlink transmission. Consequently, the interference from HUE uplink transmissions to the BUEs is negligible due to the large propagation distance and penetration loss. Furthermore, the proposed NSA-TDD eliminates the residual HAPS-induced interference through null-space precoding. As a result, the BUE performance achieved by NSA-TDD is identical to that of the BS-only configuration.
\par
Fig.~\ref{Fig:HUE_UL} and Fig.~\ref{Fig:BUE_UL} show the CDFs of the uplink SINR for HUEs and BUEs under different transmission schemes. It can be observed that the RTDD protocol achieves almost the same uplink SINR performance as the HAPS-only configuration for HUEs and the BS-only configuration for BUEs. This is because, under the RTDD protocol, the cross-tier interference is significantly attenuated by the large propagation distance and building penetration loss, making its impact on the uplink reception negligible. Consequently, the uplink performance is nearly identical to that achieved in the absence of the other network tier.
In contrast, the proposed NSA-TDD exhibits the same uplink performance as the conventional TDD scheme since the null-space projection is applied only to the HAPS downlink precoder and does not affect the uplink transmission or the MMSE receivers.

\begin{figure*}[!t]
\centering
\subfigure[Average downlink SINR of HUE.]{
\includegraphics[width=2.5in]{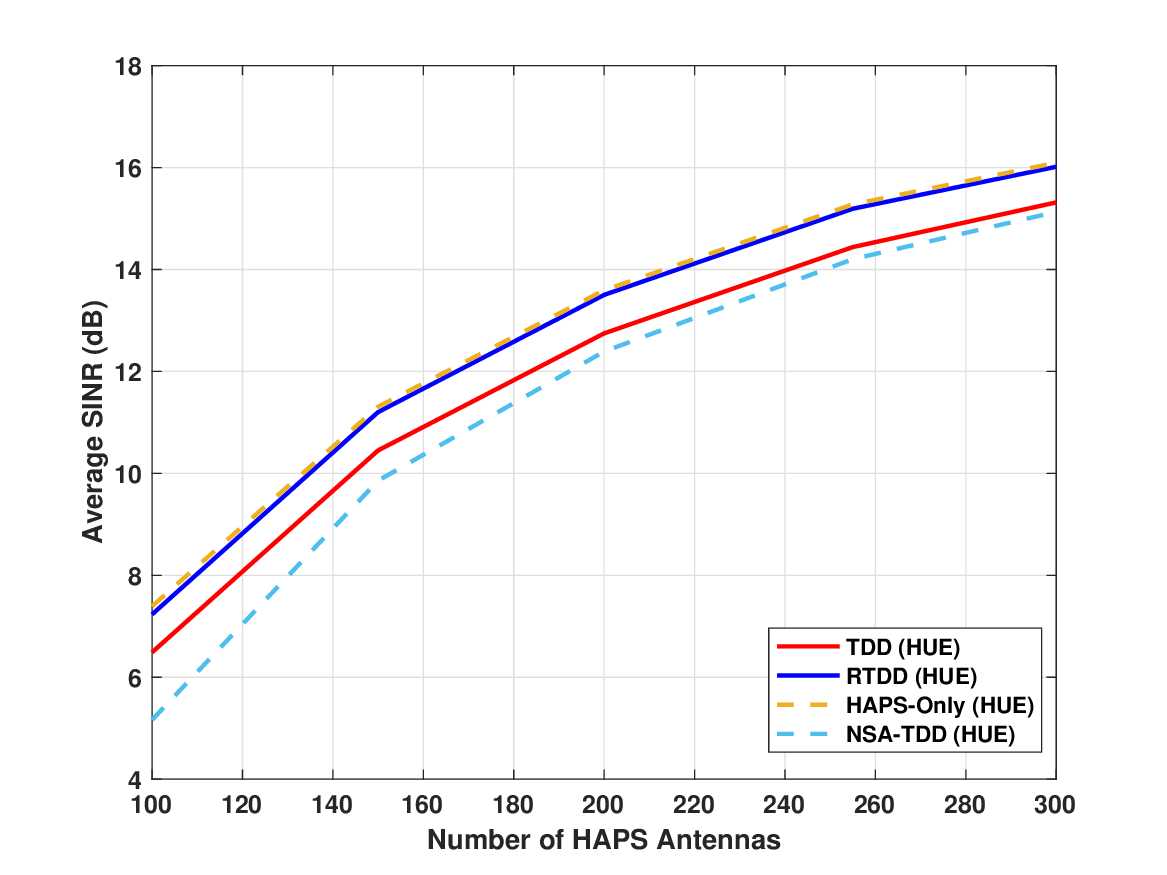}
\label{Fig:HUE_DL_HAPS_AN}
}
\hspace{-0.4cm}
\subfigure[Average downlink SINR of BUE.]{
\includegraphics[width=2.5in]{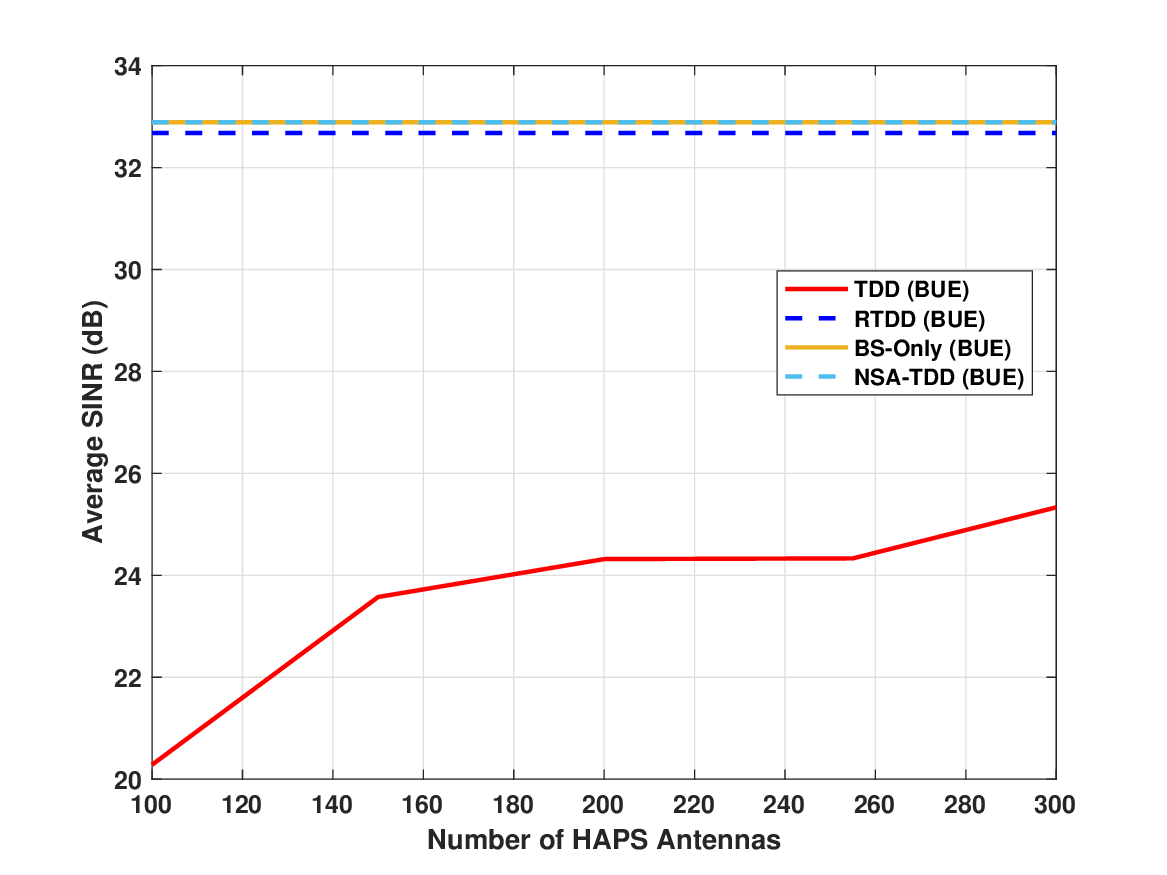}
\label{Fig:BUE_DL_HAPS_AN}
}


\subfigure[Average uplink SINR of HUE.]{
\includegraphics[width=2.5in]{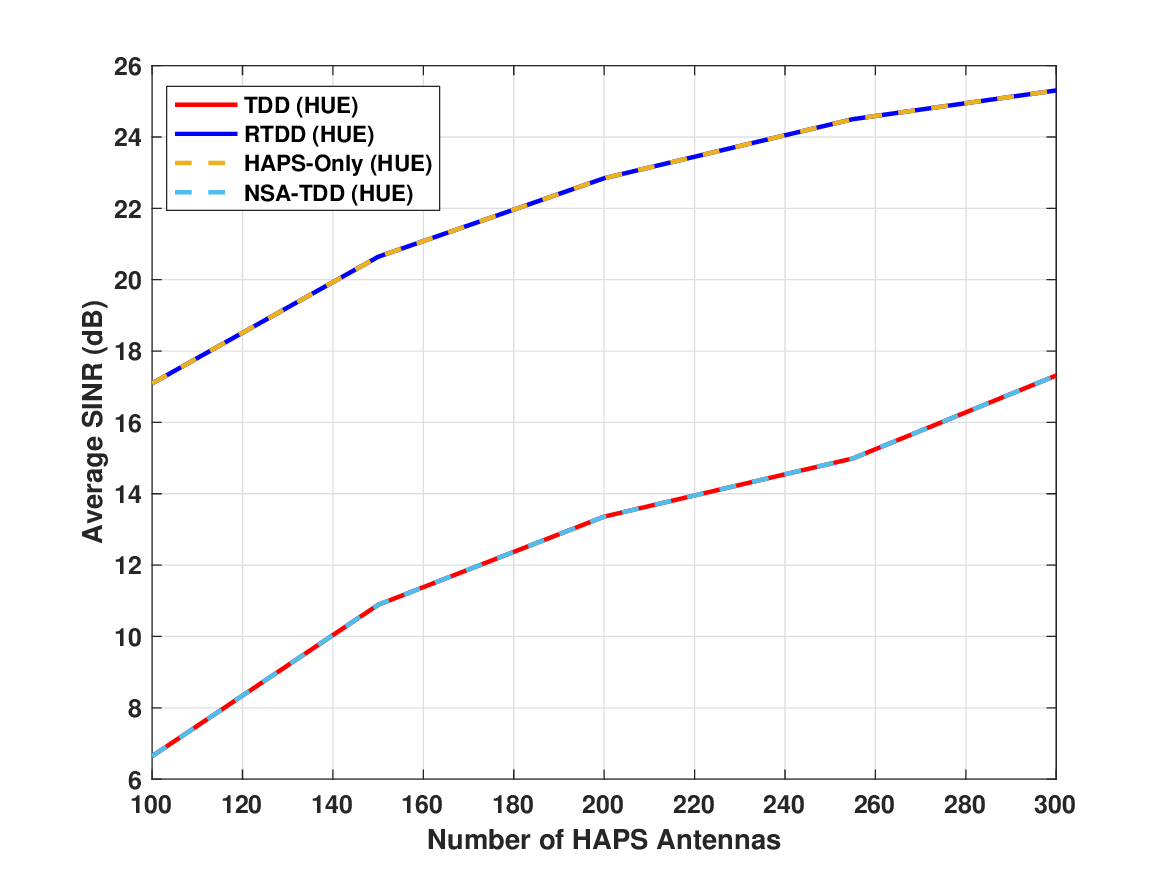}
\label{Fig:HUE_UL_HAPS_AN}
}
\hspace{-0.4cm}
\subfigure[Average uplink SINR of BUE.]{
\includegraphics[width=2.5in]{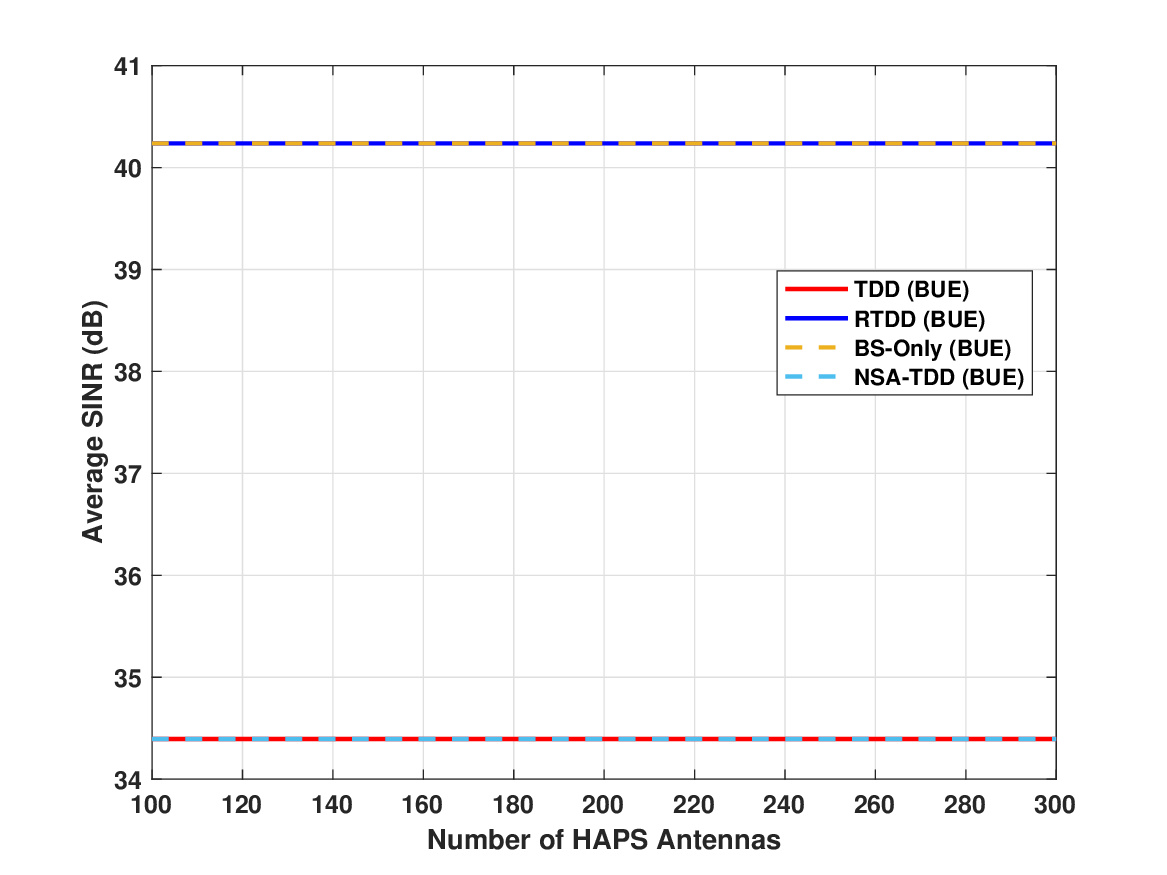}
\label{Fig:BUE_UL_HAPS_AN}
}
\caption{Average downlink and uplink SINR versus the number of HAPS antennas.}
\label{Fig:HAPS_AN}
\end{figure*}
We next investigate the impact of the number of HAPS antennas on the average SINR, as shown in Fig.~\ref{Fig:HAPS_AN}. For the downlink, Fig.~\ref{Fig:HUE_DL_HAPS_AN} shows that the average SINR of HUEs increases with the number of HAPS antennas due to the increased spatial degrees of freedom, with the proposed RTDD achieving performance close to the interference-free HUE-only benchmark. In contrast, Fig.~\ref{Fig:BUE_DL_HAPS_AN} shows that the average SINR of BUEs remains unchanged under RTDD, BS-only, and NSA-TDD, since these schemes eliminate HAPS-induced downlink interference. However, under conventional TDD, the average SINR gradually increases because a larger HAPS antenna array reduces the interference leakage toward the BUEs. For the uplink, Fig.~\ref{Fig:HUE_UL_HAPS_AN} shows that the average SINR of HUEs also increases with the number of HAPS antennas owing to the enhanced receive combining capability, whereas the average SINR of BUEs in Fig.~\ref{Fig:BUE_UL_HAPS_AN} remains unchanged since the BS reception is independent of the HAPS antenna array size.

\begin{figure*}[!t]
\centering

\subfigure[Average downlink SINR of HUE.]{
\includegraphics[width=2.5in]{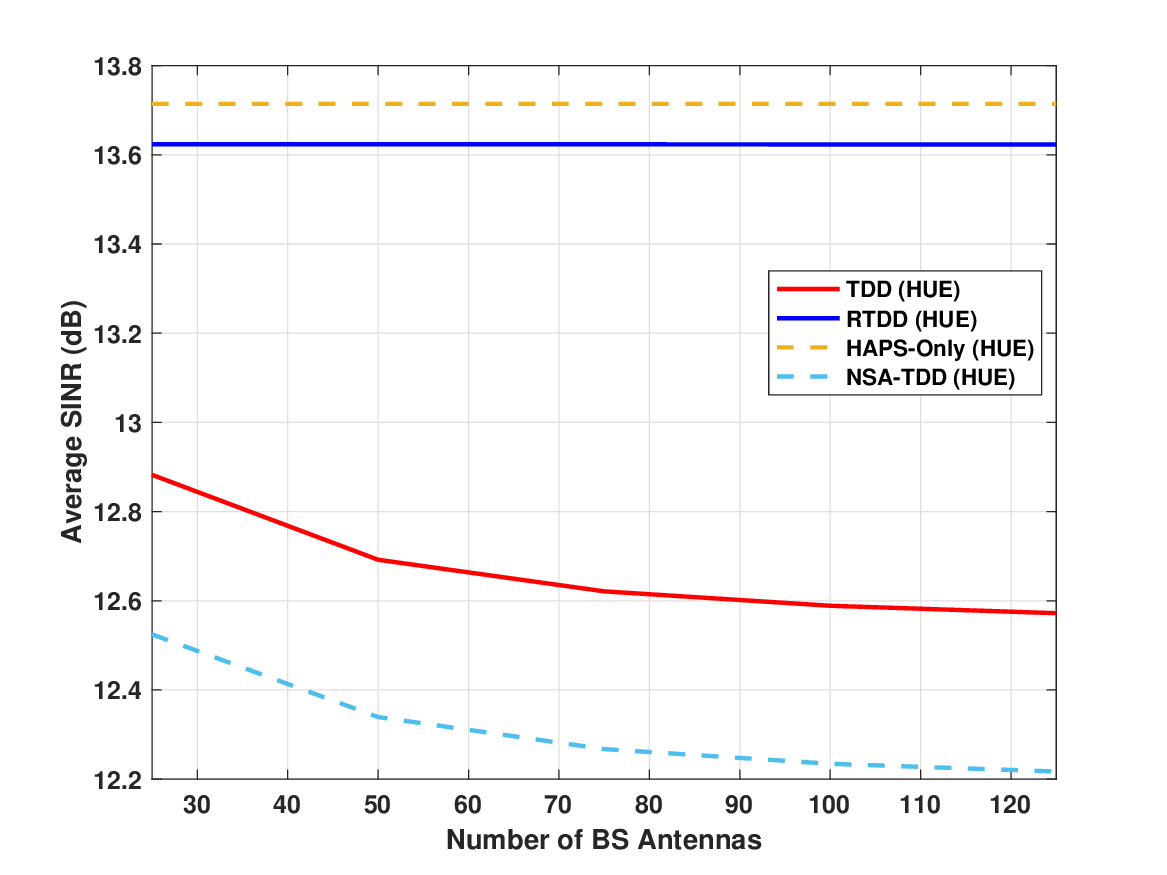}
\label{Fig:HUE_DL_BS_AN}
}
\hspace{-0.4cm}
\subfigure[Average downlink SINR of BUE.]{
\includegraphics[width=2.5in]{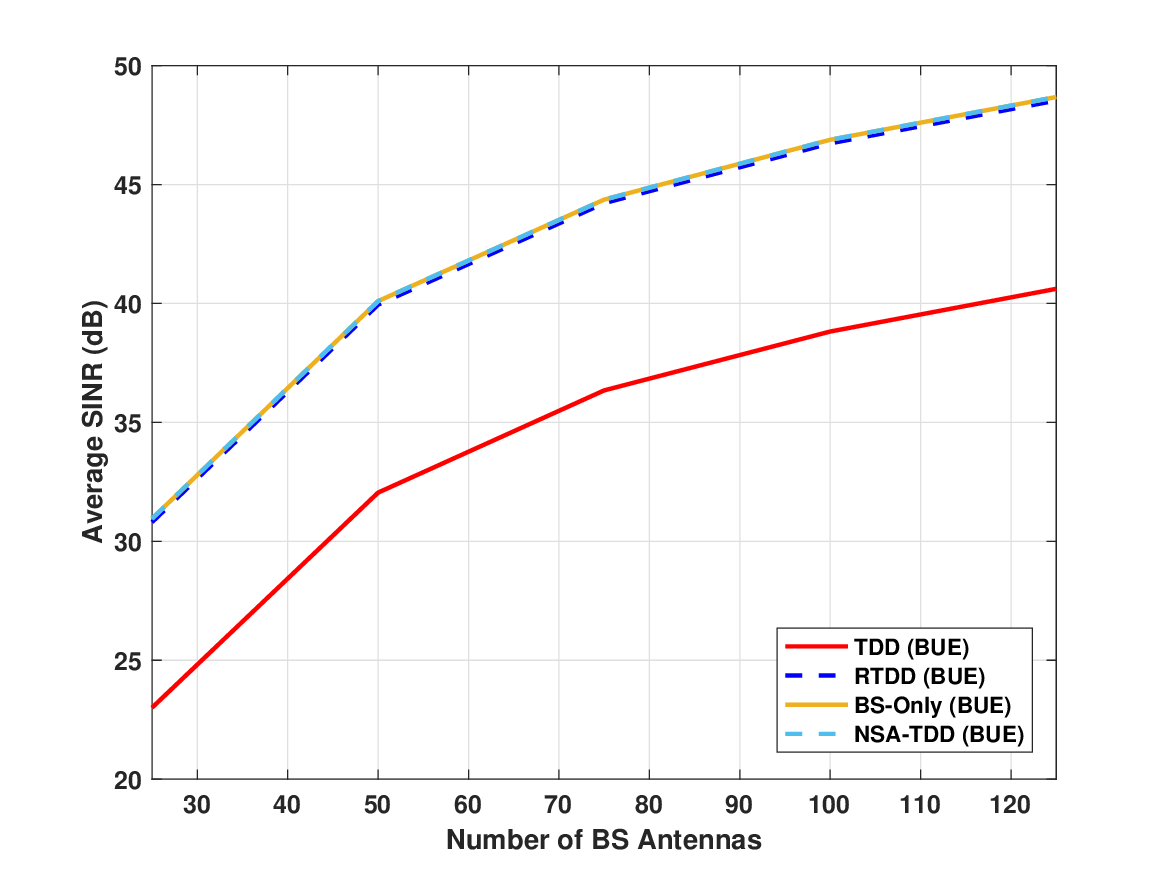}
\label{Fig:BUE_DL_BS_AN}
}

\subfigure[Average uplink SINR of HUE.]{
\includegraphics[width=2.5in]{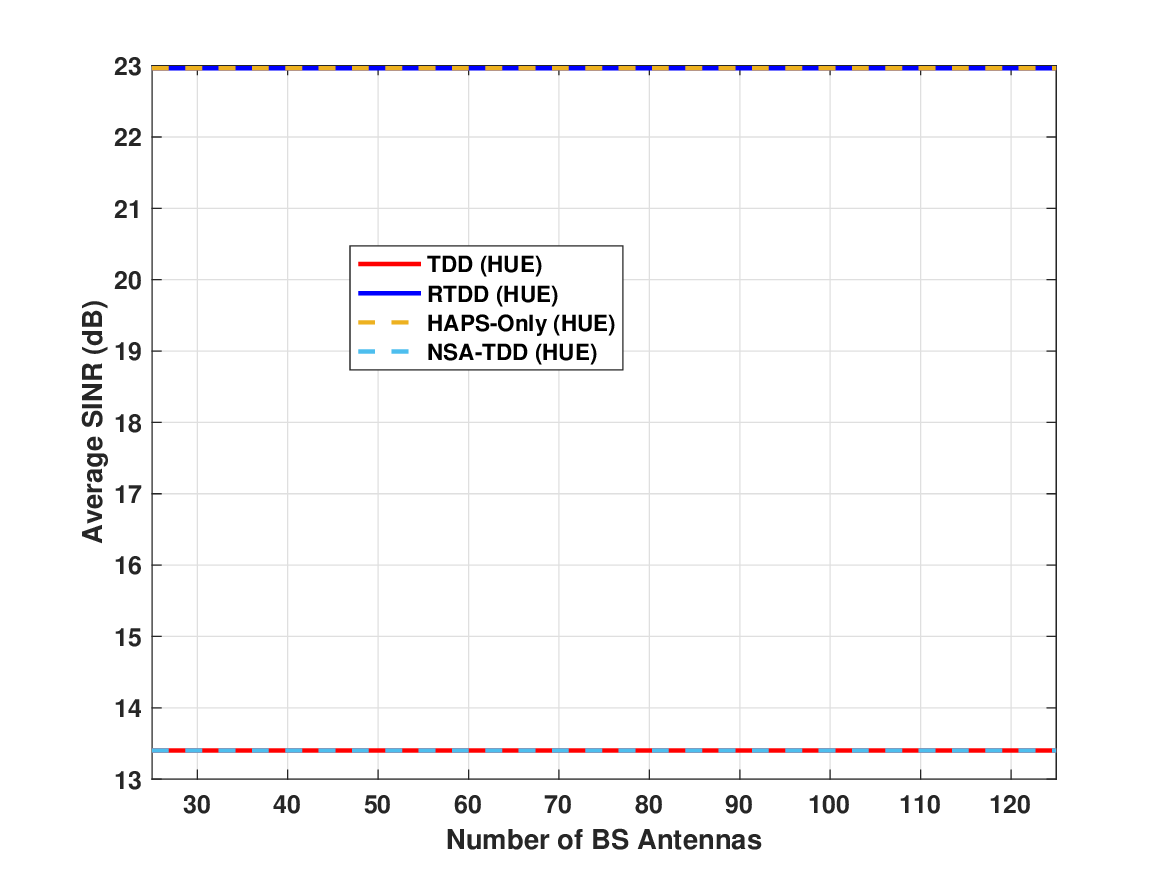}
\label{Fig:HUE_UL_BS_AN}
}
\hspace{-0.4cm}
\subfigure[Average uplink SINR of BUE.]{
\includegraphics[width=2.5in]{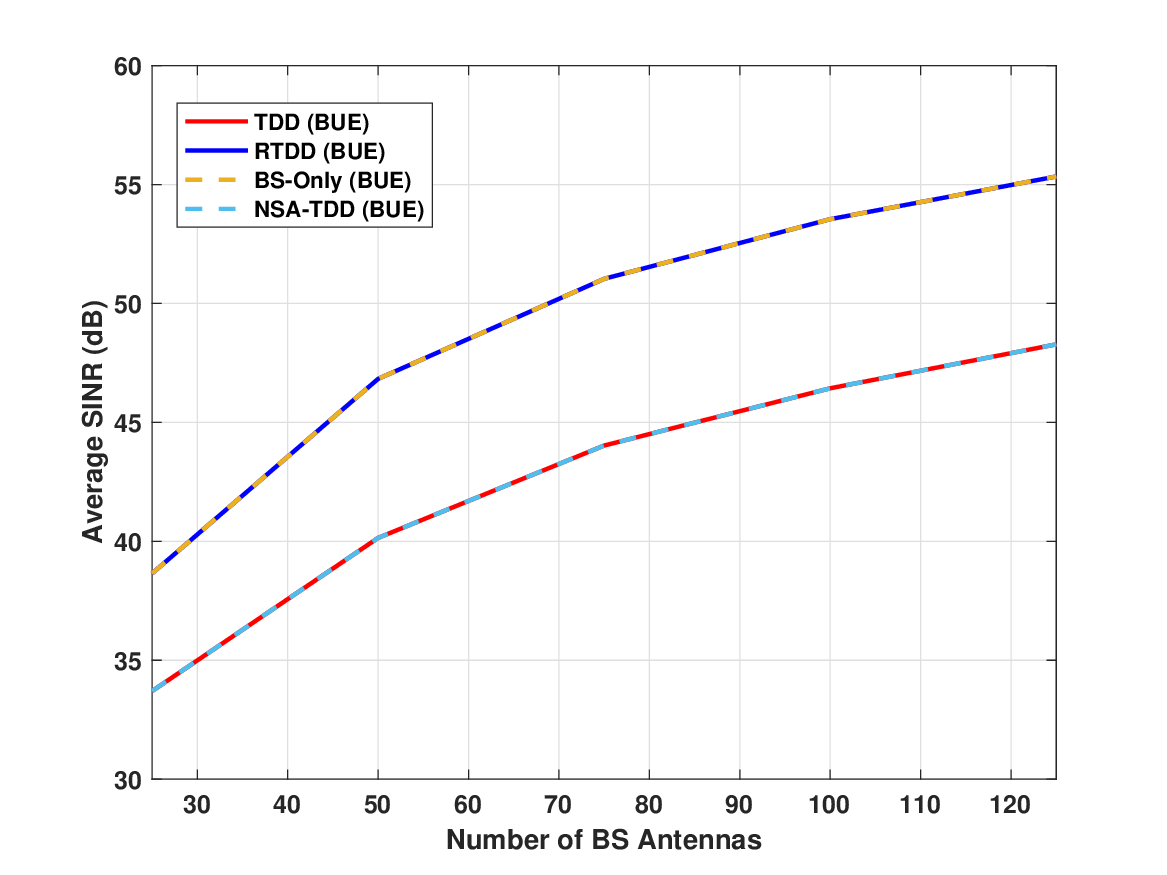}
\label{Fig:BUE_UL_BS_AN}
}

\caption{Average downlink and uplink SINR versus the number of BS antennas.}
\label{Fig:BS_AN}
\end{figure*}
Fig.~\ref{Fig:BS_AN} shows the average SINR versus the number of BS antennas. For the downlink, Fig.~\ref{Fig:HUE_DL_BS_AN} shows that the average SINR of HUEs decreases with the number of BS antennas under conventional TDD and NSA-TDD due to the increased cross-tier interference from the BS. In contrast, RTDD and the HUE-only benchmark remain unchanged since they are independent of the BS antenna configuration. As shown in Fig.~\ref{Fig:BUE_DL_BS_AN}, the average SINR of BUEs increases with the number of BS antennas owing to the enhanced beamforming capability of the BS. Nevertheless, conventional TDD still achieves the lowest performance because of the strong interference from the HAPS downlink. For the uplink, Fig.~\ref{Fig:HUE_UL_BS_AN} shows that the average SINR of HUEs remains unchanged as the number of BS antennas increases, since the HAPS reception is independent of the BS antenna configuration. In contrast, Fig.~\ref{Fig:BUE_UL_BS_AN} shows that the average SINR of BUEs increases with the number of BS antennas owing to the enhanced receive combining capability at the BS.

\subsection{Validation of the Proposed Kronecker-Based RMT Model}
In this subsection, we validate both the proposed RMT analysis and the proposed Kronecker correlation model. The accuracy of the RMT analysis is examined by comparing the derived deterministic equivalents with Monte Carlo simulations based on the Kronecker correlation model, while the accuracy of the Kronecker model is evaluated through comparison with the 3GPP channel model. 

\par
\begin{figure*}[!h]
\centering
\subfigure[CDF of downlink SINR.]{
\includegraphics[width=2.5in]{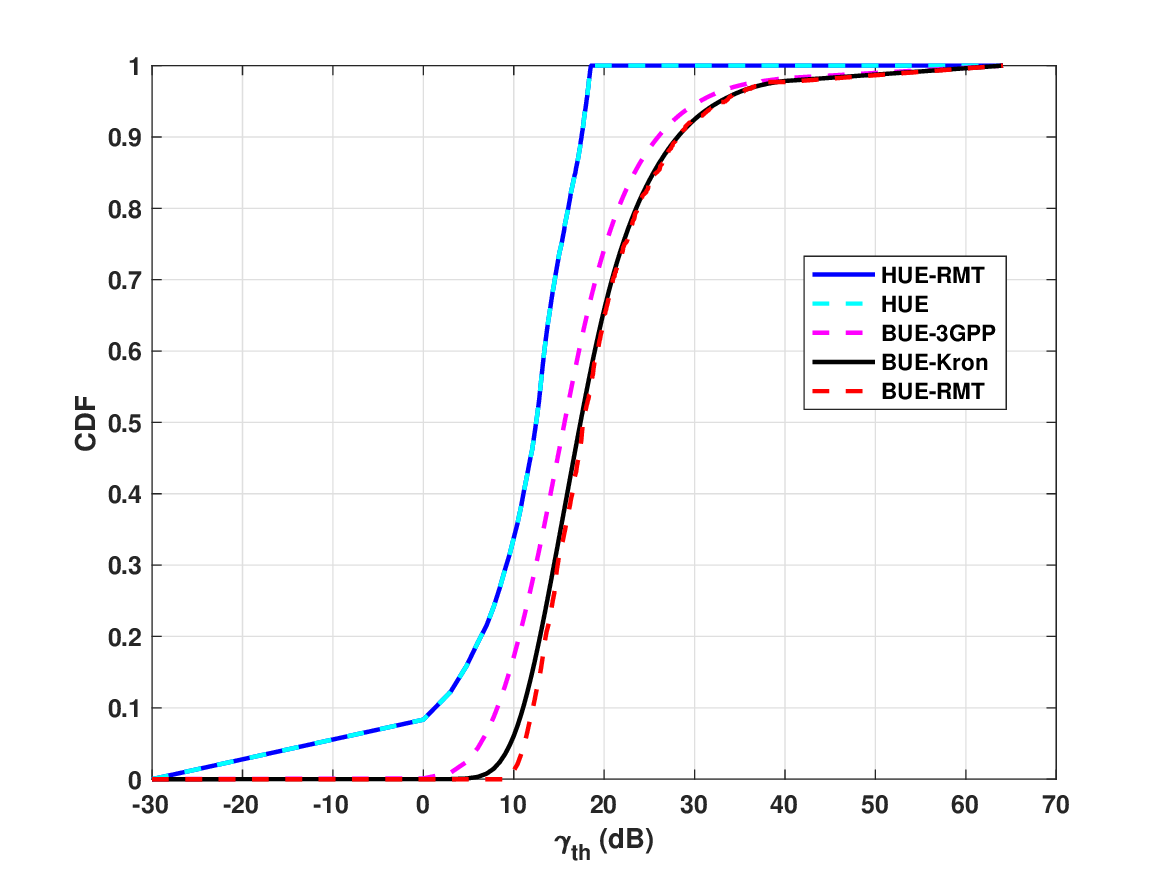}
\label{Fig:RMT_DL_N50}
}
\subfigure[CDF of uplink SINR .]{
\includegraphics[width=2.5in]{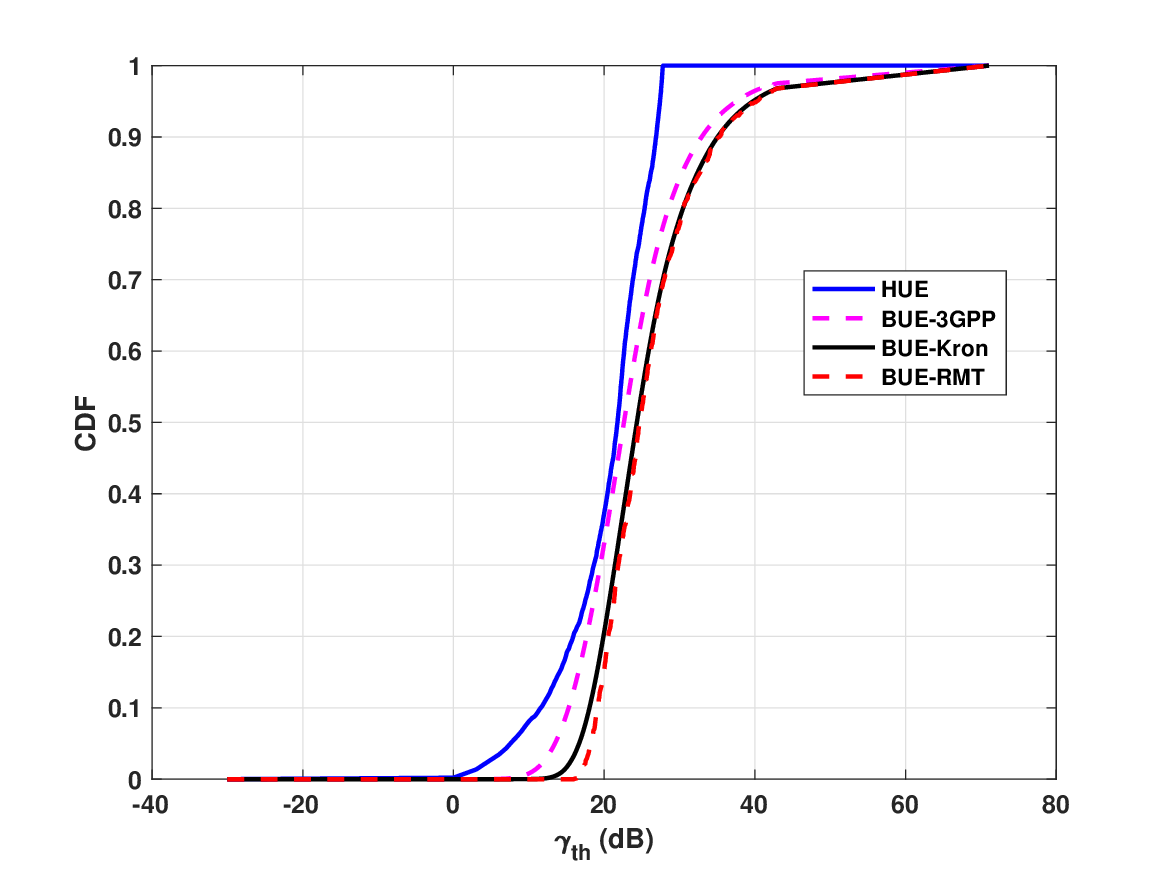}
\label{Fig:RMT_UL_N50}
}
\caption{CDFs of the downlink and uplink SINR.}
\label{Fig:RMT_N50}
\end{figure*}
Fig.~\ref{Fig:RMT_N50} compares the downlink and uplink SINR distributions obtained from the proposed RMT analysis, Monte Carlo simulations based on the proposed Kronecker correlation model, and the 3GPP channel model, where the BS is equipped with $N=50$ antennas.
Fig.~\ref{Fig:RMT_N50} shows that the proposed RMT results closely match the corresponding Monte Carlo simulations for both HUEs and BUEs. Moreover, the SINR distributions obtained from the proposed Kronecker correlation model are in good agreement with those generated by the 3GPP channel model, although a small performance gap can still be observed. This indicates that the proposed Kronecker correlation model captures the dominant spatial correlation characteristics of the practical 3GPP channel while maintaining a tractable analytical form. These results validate the accuracy of the proposed RMT analysis and demonstrate the effectiveness of the proposed Kronecker correlation model.

\begin{figure*}[!h]
\centering
\subfigure[Average downlink SINR.]{
\includegraphics[width=2.5in]{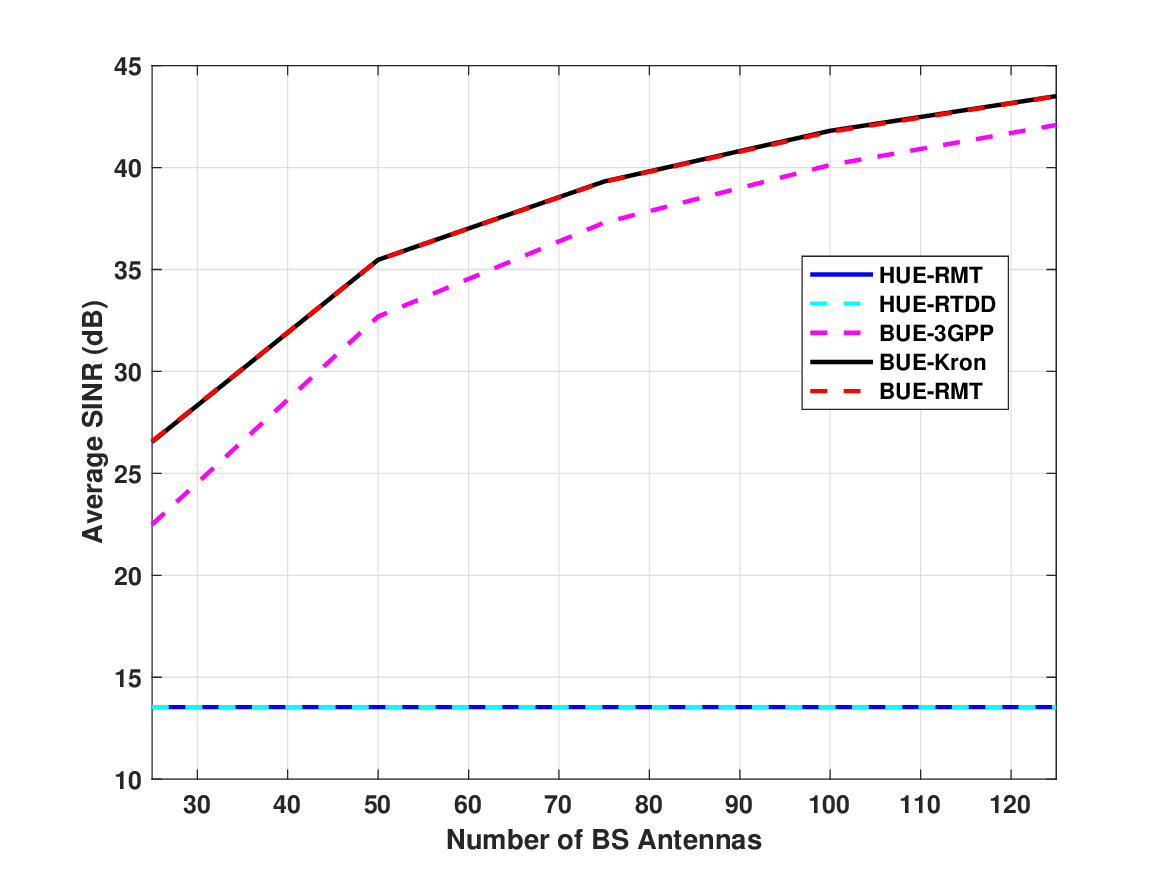}
\label{Fig:RMT_DL_BS_AN}
}
\subfigure[Average uplink SINR .]{
\includegraphics[width=2.5in]{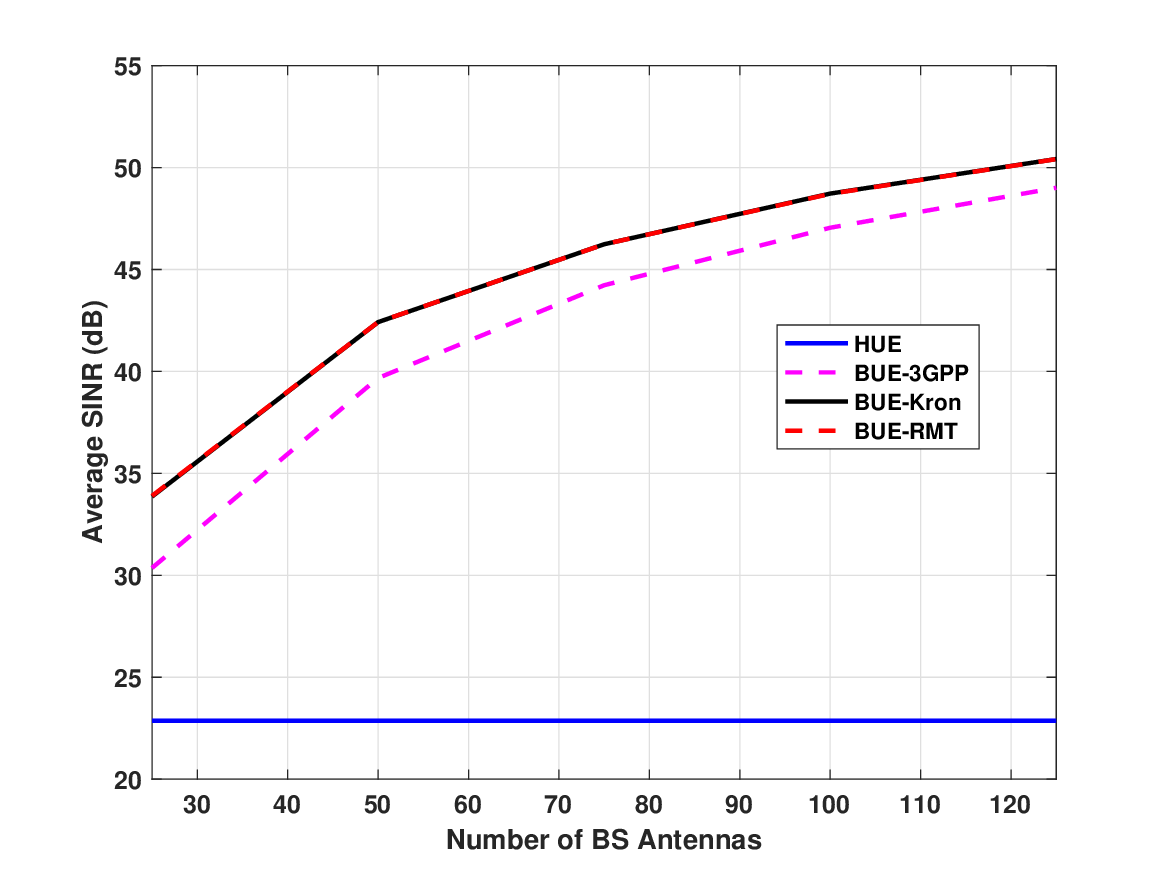}
\label{Fig:RMT_UL_BS_AN}
}
\caption{Average downlink and uplink SINR versus the number of BS antennas.}
\label{Fig:RMT_BS}
\end{figure*}
Fig.~\ref{Fig:RMT_BS} illustrates the average downlink and uplink SINR versus the number of BS antennas. It can be observed that the proposed RMT analysis closely matches the Monte Carlo simulations based on the proposed Kronecker correlation model for both HUEs and BUEs across the entire range of BS antennas, thereby validating the accuracy of the derived deterministic equivalents. Furthermore, although a small performance gap exists between the proposed Kronecker correlation model and the 3GPP channel model, both exhibit highly consistent performance trends. This demonstrates that the proposed Kronecker correlation model provides an accurate yet analytically tractable approximation of the practical 3GPP channel.

\section{Conclusion}
\label{section:Conclusions}
This paper investigated RTDD for integrated HAPS--terrestrial networks under practical 3GPP channel models. An analytical Kronecker correlation model was  developed to approximate the empirical 3GPP terrestrial BS channel, based on which a RMT framework was established to derive deterministic approximations for the uplink and downlink SINRs. Numerical results verified the accuracy of the proposed RMT analysis and demonstrated that the proposed Kronecker correlation model provides a good approximation of the practical 3GPP channel. Furthermore, the proposed RTDD framework effectively suppresses cross-tier interference and achieves near interference-free performance while preserving the channel reciprocity of TDD, making it a promising low-complexity solution for future integrated HAPS--terrestrial networks.

\appendices
\section{Supporting Lemmas and Proofs}
This section presents an auxiliary random matrix theory result that provides the deterministic equivalent of the resolvent for a sample covariance matrix with a deterministic component and correlated random columns. Building on this result, we subsequently characterize the corresponding deterministic equivalents for the ZF resolvent. For clarity and self-containment, the notation used throughout this appendix is independent of that adopted in the main text.
\subsection{Technical Lemma}
\begin{lemma}
\label{lem:stieltjes_LoS}

Given random matrices
\begin{equation}
\mathbf{\Sigma}=\mathbf{A}+\mathbf{Y}=[\boldsymbol{\xi}_1, \cdots,  \boldsymbol{\xi}_n]
\end{equation}
where $\mathbf{A}\in \mathbb{C}^{p \times n}$ is deterministic, $\mathbf{Y}=\sqrt{\frac{1}{n}}[\mathbf{B}_1\mathbf{x}_1, \mathbf{B}_2\mathbf{x}_2\cdots, \mathbf{B}_n\mathbf{x}_n ]\in \mathbb{C}^{p \times n}$ is random, $\mathbf{x_j}\in \mathbb{C}^{p \times 1}$ are i.i.d. random vectors with zero
mean and $\mathbf{B}_j \in \mathbb{C}^{p \times p}$ denotes the correlation of the $j$-th column of $\mathbf{Y}$.
\par
To investigate the random matrix $\mathbf{\Sigma}$ in a high-dimensional setting. To this
end, we make the following assumptions \cite{zhuang2026no}.
\begin{assumption}[On the Asymptotic Regime]
The dimensions of $p$ and $n$ approach infinity with 
\begin{equation}
0<\liminf _{n \rightarrow \infty} \frac{p}{n} \leq \limsup _{n \rightarrow \infty} \frac{p}{n}<\infty .
\end{equation}
\end{assumption}

\begin{assumption}[Randomness]
\label{Assump:randomness}
Let $\{X_{ij}\}_{i,j\geq 1}$ be a double array of i.i.d. complex
random variables satisfying
\begin{equation}
\mathbb E[X_{11}]=0,
\qquad
\mathbb E\!\left[|X_{11}|^2\right]=1.
\end{equation}
Moreover, there exists some $\epsilon>0$ such that
\begin{equation}
\mathbb E\!\left[|X_{11}|^{4+\epsilon}\right]<\infty.
\end{equation}
The entries of the random vectors $\{\mathbf x_j\}_{j=1}^{n}$
are drawn from this array according to
\begin{equation}
[\mathbf x_j]_i=X_{ij}.
\end{equation}
\end{assumption}

\begin{assumption}[Correlation Matrices]
\label{Assump:correlation}
Let
\begin{equation}
\boldsymbol{\Omega}_j
=
\mathbf B_j\mathbf B_j^H,
\qquad
j=1,\ldots,n.
\end{equation}
The correlation matrices satisfy
\begin{equation}
0
<
\liminf_{n\to\infty}
\min_{1\leq j\leq n}
\frac{1}{p}
\operatorname{tr}
\left(
\boldsymbol{\Omega}_j
\right)
\leq
\limsup_{n\to\infty}
\max_{1\leq j\leq n}
\left\|
\boldsymbol{\Omega}_j
\right\|
<
\infty.
\end{equation}
\end{assumption}

\begin{assumption}[on the Mean]
\label{Assump:mean}
The spectral norm of the deterministic matrix $\mathbf A$ is
uniformly bounded, i.e.,
\begin{equation}
\limsup_{n\to\infty}
\left\|
\mathbf A
\right\|
<\infty.
\end{equation}
\end{assumption}
The following lemma follows from Theorem~1 in \cite{zhuang2026no}.
Let the sample covariance matrix is defined as $\mathbf{S}=\mathbf{\Sigma}\mathbf{\Sigma}^H$, the resolvent $\mathbf{Q}(z)=(\mathbf{S}-z\mathbf{I}_p)^{-1}$ satisfy
\begin{equation}
\label{eq: resolvent}
\frac{1}{p}\operatorname{tr}\mathbf{C}[\mathbf{Q}(z)-\mathbf{\Theta}(z)]
\xrightarrow[p,n\to\infty]{a.s.} 0
\end{equation}
where $\mathbf{\Theta}$ satisfies
\begin{equation}
\begin{cases}
\delta_i(z)=\frac{\operatorname{tr} \boldsymbol{\Omega}_i \boldsymbol{\Theta}(z)}{n}, & i=1,\cdots,n, \\ \widetilde{\delta}_j(z)=[\widetilde{\boldsymbol{\Theta}}(z)]_{j, j}, & j=1,\cdots.n,
\end{cases}
\end{equation}
where 
\begin{equation}
\label{eq:F}
\mathbf{F}(z)=\left[-z\left(\mathbf{I}_p+\sum_{j=1}^n\frac{ \boldsymbol{\Omega}_j \widetilde{\delta}_j(z)}{n}\right)\right]^{-1},
\end{equation}
\begin{equation}
\label{eq:F_tilde}
\widetilde{\mathbf{F}}(z)=\operatorname{diag}\left(\frac{-1}{z\left(1+\delta_i(z)\right)} ; 1 \leq i \leq n\right),
\end{equation}
\begin{equation}
\label{eq:theta_z}
\boldsymbol{\Theta}(z)=\left(\mathbf{F}^{-1}(z)-z \mathbf{A} \widetilde{\mathbf{F}}(z) \mathbf{A}^H\right)^{-1},
\end{equation}
\begin{equation}
\widetilde{\boldsymbol{\Theta}}(z)=\left(\widetilde{\mathbf{F}}^{-1}(z)-z \mathbf{A}^H \mathbf{F}(z) \mathbf{A}\right)^{-1} .
\end{equation}
\end{lemma}

\begin{proposition}[Leave-One-Out Quadratic Form]
\label{prop:LOO_quadratic_form}
Under the assumptions of Lemma~\ref{lem:stieltjes_LoS}, let
$\mathbf{\Sigma}_u$ denote the matrix obtained by removing the
$u$-th column of $\mathbf{\Sigma}$, and define
\begin{equation}
\mathbf S_u
= \mathbf\Sigma_u\mathbf\Sigma_u^H,
\qquad
\mathbf Q_u(z)
= \left(
\mathbf S_u-z\mathbf I_p
\right)^{-1}.
\end{equation}
Then,
\begin{equation}
\boldsymbol{\xi}_u^H
\mathbf Q_u(z)
\boldsymbol{\xi}_u
- \Bigg[
\delta_u(z)+
\frac{ \left(1+\delta_u(z)\right)
\mathbf a_u^H
\boldsymbol{\Theta}(z)
\mathbf a_u
}{1+\delta_u(z)
-\mathbf a_u^H
\boldsymbol{\Theta}(z)
\mathbf a_u} \Bigg]
\xrightarrow[p,n\to\infty]{a.s.}
0,
\end{equation}
where
$
\delta_u(z)
= \frac{1}{n}
\operatorname{tr}
\left(
\boldsymbol{\Omega}_u
\boldsymbol{\Theta}(z)
\right).
$
\end{proposition}

\begin{proof}
Recall that
$
\boldsymbol{\xi}_u= \mathbf a_u+\mathbf y_u,
$
and 
$
\mathbf y_u
= \frac{1}{\sqrt n} \mathbf B_u\mathbf x_u .
$
Hence,
\begin{equation}
\begin{aligned}
\boldsymbol{\xi}_u^H
\mathbf Q_u(z)
\boldsymbol{\xi}_u
=&
\mathbf a_u^H
\mathbf Q_u(z)
\mathbf a_u
+
\mathbf y_u^H
\mathbf Q_u(z)
\mathbf y_u
\\
&+
\mathbf a_u^H
\mathbf Q_u(z)
\mathbf y_u
+
\mathbf y_u^H
\mathbf Q_u(z)
\mathbf a_u .
\end{aligned}
\end{equation}

Since $\mathbf Q_u(z)$ is independent of $\mathbf y_u$, the trace
lemma yields
\begin{equation}
\mathbf y_u^H
\mathbf Q_u(z)
\mathbf y_u
-
\frac{1}{n}
\operatorname{tr}
\left(
\boldsymbol{\Omega}_u
\mathbf Q_u(z)
\right)
\xrightarrow[p,n\to\infty]{a.s.}
0,
\end{equation}
while the cross terms satisfy
\begin{equation}
\mathbf a_u^H
\mathbf Q_u(z)
\mathbf y_u
\xrightarrow[p,n\to\infty]{a.s.}
0,
\qquad
\mathbf y_u^H
\mathbf Q_u(z)
\mathbf a_u
\xrightarrow[p,n\to\infty]{a.s.}
0.
\end{equation}
By the rank-one perturbation argument,
\begin{equation}
\frac{1}{n}
\operatorname{tr}
\left(
\boldsymbol{\Omega}_u
\mathbf Q_u(z)
\right)
-
\frac{1}{n}
\operatorname{tr}
\left(
\boldsymbol{\Omega}_u
\mathbf Q(z)
\right)
\xrightarrow[p,n\to\infty]{a.s.}
0.
\end{equation}
Together with Lemma~\ref{lem:stieltjes_LoS}, this gives
\begin{equation}
\mathbf y_u^H
\mathbf Q_u(z)
\mathbf y_u
-
\delta_u(z)
\xrightarrow[p,n\to\infty]{a.s.}
0.
\end{equation}

Moreover, for any deterministic vector with uniformly bounded norm,
the leave-one-out resolvent satisfies, which is give in the Proposition 5 in \cite{zhuang2026no} 
\begin{equation}
\mathbf a_u^H
\left[
\mathbf Q_u(z)
-
\boldsymbol{\Theta}_u(z)
\right]
\mathbf a_u
\xrightarrow[p,n\to\infty]{a.s.}
0.
\end{equation}
The leave-one-out and full deterministic equivalents are related by
\begin{equation}
\mathbf a_u^H
\boldsymbol{\Theta}(z)
\mathbf a_u
+
\frac{
\mathbf a_u^H
\boldsymbol{\Theta}_u(z)
\mathbf a_u
\,
\mathbf a_u^H
\boldsymbol{\Theta}(z)
\mathbf a_u
}{
1+\delta_u(z)
}
=
\mathbf a_u^H
\boldsymbol{\Theta}_u(z)
\mathbf a_u,
\end{equation}
which gives
\begin{equation}
\mathbf a_u^H
\boldsymbol{\Theta}_u(z)
\mathbf a_u
=
\frac{
\left(1+\delta_u(z)\right)
\mathbf a_u^H
\boldsymbol{\Theta}(z)
\mathbf a_u
}{
1+\delta_u(z)
-
\mathbf a_u^H
\boldsymbol{\Theta}(z)
\mathbf a_u
}.
\end{equation}
Combining the above relations yields the desired result.
\end{proof}

\subsection{Deterministic Equivalent of the ZF Resolvent}
\label{proof_of_ZF}
To characterize the deterministic equivalents for ZF, we consider the limiting regime in which the regularization parameter tends to zero, i.e., $z\rightarrow0$. In this regime, the fixed-point quantities and their derivatives introduced in Lemma~\ref{lem:stieltjes_LoS} require appropriate scalings to admit finite and nonvanishing limits. We therefore introduce the following assumption on these limiting quantities.
\begin{assumption}
Assume that $\lim\inf\frac{p}{n}>c_{+}>1$, and that there exists $\varphi$ and $\eta$ positive constants such that:
$$
\min_{i=1,\cdots,n}\frac{1}{n}{\rm tr} \boldsymbol{\Omega}_i(\boldsymbol{\Omega}_i+\varphi{\bf I}_p)^{-1}\geq 1+\eta
$$
\label{ass:covariance}
\end{assumption}
\begin{proposition}
\label{prop: delta}
Consider the setting of Lemma \ref{lem:stieltjes_LoS} and Assumption \ref{ass:covariance}. Then the following limits 
\begin{equation}
\underline{\delta}_i
=\lim_{t\rightarrow0}
t\delta_i(-t),
\qquad
\underline{\widetilde{\delta}}_i
=\lim_{t\rightarrow0}
\widetilde{\delta}_i(-t),
\end{equation}
exist for every $i$. Moreover, there exists a constant
$\epsilon>0$, independent of $p$, such that
\begin{equation}
\underline{\delta}_i>\epsilon,\qquad
\underline{\widetilde{\delta}}_i>\epsilon
\qquad
\forall i.
\end{equation}
\end{proposition}
\begin{proof}
The proof relies on several key steps. We first establish that, under Assumption~\ref{ass:covariance}, there exist constants $\tau>0$ and $\rho>0$, independent of $n$, such that, for all sufficiently large $n$ and every $i=1,\ldots,n$,
\begin{equation}
\boldsymbol{\Omega}_i \succeq \tau\mathbf P_i,
\qquad
\operatorname{rank}(\mathbf P_i)\geq (1+\rho)n,
\label{eq:covariance_projector}
\end{equation}
where $\mathbf P_i$ denotes the orthogonal projector onto the subspace spanned by the eigenvectors of $\boldsymbol{\Omega}_i$ associated with eigenvalues greater than or equal to $\tau$.
The second step consists in establishing the following bounds:
\begin{equation}
t\operatorname{tr}\boldsymbol{\Theta}(-t)
\geq p-n,
\qquad
t\boldsymbol{\Theta}(-t)\preceq\mathbf I_p.
\label{eq:theta_trace_bounds}
\end{equation}
We now show how these results can be combined to establish the proposition.
Combining these two results, we show that
$t\delta_i(-t)$ is uniformly bounded away from zero
as $t\downarrow 0$.
Indeed, using
$\boldsymbol{\Omega}_i\succeq\tau\mathbf P_i$
and the positive semidefiniteness of
$\boldsymbol{\Theta}(-t)$, we obtain
\begin{align}
t\delta_i(-t)
&=\frac{t}{n}\operatorname{tr}
\bigl(\boldsymbol{\Omega}_i
\boldsymbol{\Theta}(-t)\bigr)
\nonumber\\
&\geq
\frac{t\tau}{n}
\operatorname{tr}
\bigl(\mathbf P_i
\boldsymbol{\Theta}(-t)\bigr)
\nonumber\\
&=
\frac{t\tau}{n}
\operatorname{tr}\boldsymbol{\Theta}(-t)
\nonumber\\
&\quad-
\frac{t\tau}{n}
\operatorname{tr}
\bigl((\mathbf I_p-\mathbf P_i)
\boldsymbol{\Theta}(-t)\bigr).
\label{eq:delta_lower_bound_step1}
\end{align}
Since
$t\boldsymbol{\Theta}(-t)\preceq\mathbf I_p$,
we have
\begin{align}
&t\operatorname{tr}
\bigl((\mathbf I_p-\mathbf P_i)
\boldsymbol{\Theta}(-t)\bigr)
\nonumber\\
&\qquad\leq
\operatorname{tr}(\mathbf I_p-\mathbf P_i)
\nonumber\\
&\qquad=
p-\operatorname{rank}(\mathbf P_i).
\label{eq:projector_complement_bound}
\end{align}
Consequently, using
\eqref{eq:theta_trace_bounds} and
\eqref{eq:covariance_projector}, we obtain
\begin{align}
t\delta_i(-t)
&\geq
\frac{\tau}{n}
\left(
p-n-
\bigl(p-\operatorname{rank}(\mathbf P_i)\bigr)
\right)
\nonumber\\
&=
\frac{\tau}{n}
\left(\operatorname{rank}(\mathbf P_i)-n\right)
\nonumber\\
&\geq \tau\rho.
\label{eq:delta_uniform_lower_bound}
\end{align}
Therefore,
\begin{equation}
\inf_{t>0}\,
\min_{1\leq i\leq n}
t\delta_i(-t)
\geq\tau\rho>0.
\label{eq:delta_positive_lower_bound}
\end{equation}
In particular, $t\delta_i(-t)$ remains uniformly
bounded away from zero as $t\downarrow0$.
Using \cite{zhuang2026no}, the function
$z\mapsto\delta_i(z)$ is the Stieltjes transform
of a finite positive measure $\mu_i$ supported
on $[0,\infty)$. Consequently,
\begin{equation}
t\delta_i(-t)
=
\int_{[0,\infty)}
\frac{t}{t+\lambda}\,\mu_i(d\lambda).
\end{equation}
Since the mapping $t\mapsto t/(t+\lambda)$ is
nondecreasing for every $\lambda\geq0$,
$t\delta_i(-t)$ is also nondecreasing with
respect to $t$. Moreover, using
$t\boldsymbol{\Theta}(-t)\preceq\mathbf I_p$,
we obtain
\begin{equation}
t\delta_i(-t)
\leq
\frac{1}{n}\operatorname{tr}
\boldsymbol{\Omega}_i<\infty.
\end{equation}
Therefore, by monotonicity and boundedness,
$t\delta_i(-t)$ admits a finite limit as
$t\downarrow0$. Combining this result with
the previously established lower bound yields
\begin{equation}
\lim_{t\downarrow0}t\delta_i(-t)
=d_i\geq\tau\rho>0.
\end{equation}
We next establish the convergence of
$\widetilde{\delta}_i(-t)$ as $t\downarrow0$.
By Proposition~1 in \cite{zhuang2026no},
$\widetilde{\delta}_i(z)$ is the Stieltjes
transform of a nonnegative measure
$\widetilde{\mu}_i$ supported on $[0,\infty)$.
Consequently,
\begin{equation}
\widetilde{\delta}_i(-t)
=
\int_{[0,\infty)}
\frac{1}{t+\lambda}
\,\widetilde{\mu}_i(d\lambda).
\end{equation}
Since the mapping $t\mapsto 1/(t+\lambda)$
is nonincreasing for every $\lambda\geq0$,
$\widetilde{\delta}_i(-t)$ is nonincreasing
with respect to $t$.

Moreover, using the fixed-point equations,
we have
\begin{align}
0<\widetilde{\delta}_i(-t)
&\leq
\frac{1}{t(1+\delta_i(-t))}
\nonumber\\
&\leq
\frac{1}{\tau\rho}.
\end{align}
Therefore, by monotonicity and boundedness,
$\widetilde{\delta}_i(-t)$ admits a finite
limit as $t\downarrow0$, namely,
\begin{equation}
\widetilde d_i
:=
\lim_{t\downarrow0}
\widetilde{\delta}_i(-t)
\leq
\frac{1}{\tau\rho}.
\end{equation}
Using the inequality
$[\mathbf M^{-1}]_{ii}\geq 1/M_{ii}$,
which holds for any Hermitian positive definite
matrix $\mathbf M$, we obtain
\begin{align}
\widetilde{\delta}_i(-t)
&=[\widetilde{\boldsymbol{\Theta}}(-t)]_{ii}
\nonumber\\
&\geq
\frac{1}{
[\widetilde{\boldsymbol{\Theta}}^{-1}(-t)]_{ii}
}
\nonumber\\
&=
\frac{1}{
t(1+\delta_i(-t))
+\mathbf a_i^H\mathbf F\mathbf a_i
}
\nonumber\\
&\geq
\frac{1}{
t+\frac{1}{n}\operatorname{tr}
\boldsymbol{\Omega}_i+\|\mathbf a_i\|^2
},
\end{align}
where the last inequality follows from
$\mathbf F\preceq\mathbf I_p$ and
$t\delta_i(-t)\leq
n^{-1}\operatorname{tr}\boldsymbol{\Omega}_i$.\\
\underline{Proof of \eqref{eq:covariance_projector}.} 
We now prove \eqref{eq:covariance_projector}.
Let $\lambda_{i1},\ldots,\lambda_{ip}$ denote
the eigenvalues of $\boldsymbol{\Omega}_i$.
By Assumption~\ref{ass:covariance},
\begin{equation}
\frac{1}{n}\sum_{k=1}^{p}
\frac{\lambda_{ik}}{\lambda_{ik}+\varphi}
\geq 1+\eta.
\end{equation}

Choose $\tau>0$ sufficiently small such that
\begin{equation}
\frac{\tau}{\tau+\varphi}\frac{p}{n}
\leq \frac{\eta}{2}.
\end{equation}
Since
$\lambda/(\lambda+\varphi)
\leq \tau/(\tau+\varphi)$
for $\lambda<\tau$, we obtain
\begin{align}
1+\eta
&\leq
\frac{1}{n}\sum_{k=1}^{p}
\frac{\lambda_{ik}}{\lambda_{ik}+\varphi}
\nonumber\\
&\leq
\frac{\operatorname{rank}(\mathbf P_i)}{n}
+
\frac{\tau}{\tau+\varphi}
\frac{p-\operatorname{rank}(\mathbf P_i)}{n}
\nonumber\\
&\leq
\frac{\operatorname{rank}(\mathbf P_i)}{n}
+\frac{\eta}{2}.
\end{align}
Consequently,
\begin{equation}
\operatorname{rank}(\mathbf P_i)
\geq \left(1+\frac{\eta}{2}\right)n.
\end{equation}

Setting $\rho=\eta/2$ and recalling that
$\mathbf P_i$ is the orthogonal projector
onto the eigenspace of
$\boldsymbol{\Omega}_i$ associated with
eigenvalues greater than or equal to $\tau$,
we conclude that
\begin{equation}
\boldsymbol{\Omega}_i\succeq\tau\mathbf P_i,
\qquad
\operatorname{rank}(\mathbf P_i)
\geq(1+\rho)n,
\end{equation}
which establishes \eqref{eq:covariance_projector}.

\underline{Proof of \eqref{eq:theta_trace_bounds}}
We now prove \eqref{eq:theta_trace_bounds}.
From the fixed-point equations, we have
\begin{equation}
\boldsymbol{\Theta}(-t)^{-1}
=
t\mathbf I_p+
\frac{1}{n}\sum_{i=1}^{n}
\boldsymbol{\Omega}_it\tilde{\delta}_i(-t)
+
t\mathbf A\widetilde{\mathbf F}
\mathbf A^H,
\end{equation}
where the last term is positive semidefinite.
Consequently,
\begin{equation}
\boldsymbol{\Theta}(-t)^{-1}
\succeq t\mathbf I_p,
\end{equation}
which yields
\begin{equation}
t\boldsymbol{\Theta}(-t)
\preceq\mathbf I_p.
\end{equation}

To establish the first inequality, we use the
corresponding companion fixed-point equations.
Writing
\begin{equation}
\boldsymbol{\Theta}(-t)
=
\left[
t\mathbf I_p+
\frac{1}{n}\sum_{i=1}^{n}t\tilde{\delta}_i(-t)\boldsymbol{\Omega}_i
+t\mathbf A
\widetilde{\mathbf{F}}
\mathbf A^H
\right]^{-1},
\end{equation}
and taking the trace after multiplication by
$\boldsymbol{\Theta}(-t)^{-1}$ gives
\begin{align}
p
&=
t\operatorname{tr}\boldsymbol{\Theta}(-t)
+
\sum_{i=1}^{n}
t\tilde{\delta}_i(-t)\delta_i(-t)
+
t\operatorname{tr}\left[
\mathbf A\widetilde{\mathbf F}
\mathbf A^H
\boldsymbol{\Theta}(-t)
\right].\label{eq:first_eq}
\end{align}

Similarly, multiplying the fixed-point equation defining
$\widetilde{\boldsymbol{\Theta}}(-t)$ by
$\widetilde{\boldsymbol{\Theta}}(-t)^{-1}$ and taking
the trace yields
\begin{align}
n
&=
t\operatorname{tr}\widetilde{\boldsymbol{\Theta}}(-t)
+
\sum_{i=1}^{n}
t\delta_i(-t)\tilde{\delta}_i(-t)+
t\operatorname{tr}\left[
\mathbf A^H\mathbf F\mathbf A
\widetilde{\boldsymbol{\Theta}}(-t)
\right].
\label{eq:trace_tilde_theta}
\end{align}
Using the Woodbury matrix identity
\begin{equation}
(\mathbf X+\mathbf U\mathbf V)^{-1}\mathbf U
=
\mathbf X^{-1}\mathbf U
(\mathbf I+\mathbf V\mathbf X^{-1}\mathbf U)^{-1},
\end{equation}
we obtain
$$
\boldsymbol{\Theta}(-t)\mathbf A\widetilde{\mathbf F}=\mathbf F\mathbf A\widetilde{\boldsymbol{\Theta}}(-t)
$$
and therefore
\begin{align}
&\operatorname{tr}\left[
\mathbf A\widetilde{\mathbf F}(-t)
\mathbf A^H\boldsymbol{\Theta}(-t)
\right]
\nonumber\\
&\quad=
\operatorname{tr}\left[
\mathbf A^H\mathbf F(-t)\mathbf A
\widetilde{\boldsymbol{\Theta}}(-t)
\right].
\end{align}

Subtracting \eqref{eq:trace_tilde_theta}
from \eqref{eq:first_eq}, we obtain
\begin{equation}
t\operatorname{tr}\boldsymbol{\Theta}(-t)
-
t\operatorname{tr}
\widetilde{\boldsymbol{\Theta}}(-t)
=p-n.
\end{equation}
Since
$\widetilde{\boldsymbol{\Theta}}(-t)\succeq\mathbf 0$,
it follows that
\begin{equation}
t\operatorname{tr}\boldsymbol{\Theta}(-t)
\geq p-n.
\end{equation}
\end{proof}

\begin{theorem}
\label{thm:zf_resolvent}
Under the assumptions of Lemma~\ref{lem:stieltjes_LoS} and Proposition~\ref{prop: delta}, 
let $\mathbf C\in\mathbb C^{p\times p}$ be any deterministic matrix with uniformly bounded spectral norm. Then,
\begin{equation}
\frac{1}{p}\operatorname{tr}\mathbf{C}[\underline{\mathbf{Q}}-\underline{\mathbf{\Theta}}]
\xrightarrow[p,n\to\infty]{a.s.} 0
\end{equation}
where
\[ \underline{\mathbf Q} = \lim_{z\rightarrow0}(-z)\mathbf Q(z), \]
and
$\underline{\mathbf\Theta}$ is characterized by the following fixed-point equations:
\begin{equation}
\begin{cases}
\underline{\delta}_i
=\dfrac{\operatorname{tr}
\boldsymbol{\Omega}_i
\underline{\boldsymbol{\Theta}}}{n},
&i\in[n],
\\[1mm]
\underline{\widetilde{\delta}}_i
=[\underline{\widetilde{\boldsymbol{\Theta}}}]_{i,i},
&
i\in[n].
\end{cases}
\end{equation}
where 
\begin{equation}
\underline{\mathbf F}
=\left( \mathbf I_p + \frac1n \sum_{j=1}^{n}
\boldsymbol{\Omega}_j\underline{\widetilde{\delta }}_j \right)^{-1},
\end{equation}
\begin{equation}
\underline{\widetilde{\mathbf F}}
=
\operatorname{diag}
\left(
\frac1{\underline{\delta}_i}
\right),
\end{equation}

\begin{equation}
\underline{\boldsymbol{\Theta}}
=\left(
\underline{\mathbf F}^{-1}
+
\mathbf A
\underline{\widetilde{\mathbf F}}
\mathbf A^H
\right)^{-1},
\end{equation}

\begin{equation}
\underline{\widetilde{\boldsymbol{\Theta}}}
= \left(
\underline{\widetilde{\mathbf F}}^{-1}
+
\mathbf A^H
\underline{\mathbf F}
\mathbf A
\right)^{-1}.
\end{equation}
\end{theorem}
\begin{proof}
Following an argument similar to that in \cite{wagner2012large}, we first derive
the limiting fixed-point equations as the regularization parameter approaches
zero. 
Define
\begin{equation}
\underline{\boldsymbol{\Theta}}
=
\lim_{z\rightarrow0}
(-z)\boldsymbol{\Theta}(z),
\qquad
\underline{\widetilde{\boldsymbol{\Theta}}}
=
\lim_{z\rightarrow0}
\widetilde{\boldsymbol{\Theta}}(z).
\end{equation}
Applying the limiting scalings in proposition \ref{prop: delta} 
to the fixed-point equations in Lemma~\ref{lem:stieltjes_LoS}
and letting $z\rightarrow0$ yields
\begin{equation}
\begin{aligned}
\underline{\delta}_i
&=
\lim_{z\rightarrow0}
\frac1n
\operatorname{tr}
\left(
-z
\boldsymbol{\Omega}_i
\boldsymbol{\Theta}(z)
\right)
\\
&=
\lim_{z\rightarrow0}
\frac1n
\operatorname{tr}
\left[
\boldsymbol{\Omega}_i
\left(
\frac{\mathbf F^{-1}(z)}{-z}
+
\mathbf A
\widetilde{\mathbf F}(z)
\mathbf A^H
\right)^{-1}
\right]
\\
&=
\frac1n
\operatorname{tr}
\left(
\boldsymbol{\Omega}_i
\underline{\boldsymbol{\Theta}}
\right),
\end{aligned}
\end{equation}
and
\begin{equation}
\underline{\widetilde{\delta}}_i
=
[\underline{\widetilde{\boldsymbol{\Theta}}}]_{ii}.
\end{equation}
Taking the corresponding limits in
\eqref{eq:F} and \eqref{eq:F_tilde} yields
\begin{equation}
\underline{\mathbf{F}}=\left(\mathbf{I}_N+\sum_{j=1}^n \frac{\boldsymbol{\Omega}_j \underline{\widetilde{\delta}}_j}{n}\right)^{-1},
\end{equation}

\begin{equation}
\underline{\widetilde{\mathbf{F}}}=\operatorname{diag}\left(\frac{1}{\underline{\delta}_i} ; 1 \leq i \leq n\right),
\end{equation}

\begin{equation}
\underline{\boldsymbol{\Theta}}=\left(\underline{\mathbf{F}}^{-1}+\mathbf{A} \underline{\widetilde{\mathbf{F}}}\mathbf{A}^H\right)^{-1},
\end{equation}

\begin{equation}
\underline{\widetilde{\boldsymbol{\Theta}}}=\left(\underline{\widetilde{\mathbf{F}}}^{-1}+\mathbf{A}^H \underline{\mathbf{F}} \mathbf{A}\right)^{-1} 
\end{equation}
\end{proof}

\begin{proposition}[Leave-One-Out Quadratic Form for ZF]
\label{prop:LOO_quadratic_form_ZF}
Under the assumptions of Lemma~\ref{lem:stieltjes_LoS} and
Proposition~\ref{prop: delta}, define
$
\underline{\mathbf Q}_u
=
\lim_{z\to0}
(-z)\mathbf Q_u(z),
$
and
$
\underline{\delta}_u
=
\frac{1}{n}
\operatorname{tr}
\left(
\boldsymbol{\Omega}_u
\underline{\boldsymbol{\Theta}}
\right).
$
Then,
\begin{equation}
\boldsymbol{\xi}_u^H
\underline{\mathbf Q}_u
\boldsymbol{\xi}_u
-
(\underline{\delta}_u
+
\frac{
\underline{\delta}_u
\mathbf a_u^H
\underline{\boldsymbol{\Theta}}
\mathbf a_u
}{
\underline{\delta}_u
-
\mathbf a_u^H
\underline{\boldsymbol{\Theta}}
\mathbf a_u
})
\xrightarrow[p,n\to\infty]{a.s.}
0.
\end{equation}
\end{proposition}

\begin{proof}
Recall From Proposition~\ref{prop:LOO_quadratic_form}, we have 
\begin{equation}
\begin{aligned}
\boldsymbol{\xi}_u^H
\mathbf Q_u(z)
\boldsymbol{\xi}_u
- \Bigg[
\delta_u(z) +
\frac{ \left(1+\delta_u(z)\right)
\mathbf a_u^H
\boldsymbol{\Theta}(z)
\mathbf a_u
}{ 1+\delta_u(z)
-
\mathbf a_u^H
\boldsymbol{\Theta}(z)
\mathbf a_u }
\Bigg]
\xrightarrow[p,n\to\infty]{a.s.} 0.
\end{aligned}
\end{equation}
Multiplying by $-z$ and letting $z\to0$, the first term satisfies
\begin{equation}
\lim_{z\to0}
(-z)\delta_u(z)
=
\underline{\delta}_u.
\end{equation}
Moreover,
\begin{equation}
\lim_{z\to0}
(-z)
\mathbf a_u^H
\boldsymbol{\Theta}(z)
\mathbf a_u
=
\mathbf a_u^H
\underline{\boldsymbol{\Theta}}
\mathbf a_u.
\end{equation}
Hence,
\begin{equation}
\lim_{z\to0}
(-z)
\frac{
\left(1+\delta_u(z)\right)
\mathbf a_u^H
\boldsymbol{\Theta}(z)
\mathbf a_u
}{
1+\delta_u(z)
-
\mathbf a_u^H
\boldsymbol{\Theta}(z)
\mathbf a_u
}
=
\frac{
\underline{\delta}_u
\mathbf a_u^H
\underline{\boldsymbol{\Theta}}
\mathbf a_u
}{
\underline{\delta}_u
-
\mathbf a_u^H
\underline{\boldsymbol{\Theta}}
\mathbf a_u
}.
\end{equation}
Therefore,
\begin{equation}
\boldsymbol{\xi}_u^H
\underline{\mathbf Q}_u
\boldsymbol{\xi}_u
\asymp
\underline{\delta}_u
+
\frac{
\underline{\delta}_u
\mathbf a_u^H
\underline{\boldsymbol{\Theta}}
\mathbf a_u
}{
\underline{\delta}_u
-
\mathbf a_u^H
\underline{\boldsymbol{\Theta}}
\mathbf a_u
}
\end{equation}
This completes the proof.
\end{proof}

\section{Proof of Theorem \ref{thm:BUE_DL_SINR}}
\label{app:proof_of_BUE_DL_SINR}
We now analyze the desired signal power
$\frac{P^{(\rm BS)}}{U}|(\hat{\mathbf h}_u^{(\rm BS)})^H\mathbf w_u^{(\rm BS)}|^2.$
The derivation relies on the limiting deterministic equivalents for ZF precoding established in Theorems~\ref{thm:zf_resolvent}. For notational simplicity, the superscript $(\rm BS)$ is omitted throughout this subsection.
The normalized channel vector is given by
\begin{equation}
\hat{\mathbf h}_u
= \sqrt{\rho_u} \left( \sqrt{\frac{1}{1+\kappa_u}}
\mathbf h_u^{\rm NLoS}
+ \sqrt{\frac{\kappa_u}{1+\kappa_u}} \mathbf h_u^{\rm LoS} \right),
\end{equation}
where
$
\mathbf h_u^{\rm NLoS}
=\mathbf R_u^{1/2}\mathbf z_u,
$ and
$\mathbf z_u\sim\mathcal{CN}(\mathbf0,\mathbf I).$
\par
Since
$
\mathbf w_u=
\frac{\tilde{\mathbf w}_u} {\|\tilde{\mathbf w}_u\|},
$
the desired signal power is expressed as
\begin{equation}
\label{eq:signal}
\frac{1}{U}
\left|
\hat{\mathbf h}_u^H
\mathbf w_u
\right|^2
=\frac{|\frac{1}{U}
\hat{\mathbf h}_u^H
\tilde{\mathbf w}_u|^2
}{
\frac{1}{U}
\tilde{\mathbf w}_u^H
\tilde{\mathbf w}_u
}.
\end{equation}
By exploiting the relationship between ZF precoding and the regularized inverse, we obtain
\begin{equation}
\begin{aligned}
|\frac{1}{U}\hat{\mathbf{h}}_u^H\tilde{\mathbf{w}}_u|^2=&\lim_{\alpha\rightarrow 0}|\frac{1}{U}\hat{\mathbf{h}}_u^H\hat{\mathbf{W}}\mathbf{\mathbf{h}}_u|^2\\
=&\lim_{\alpha\rightarrow 0}|\frac{\frac{1}{U}\hat{\mathbf{h}}_u^H\hat{\mathbf{W}}_u\hat{\mathbf{h}}_u}{1+\frac{1}{U}\hat{\mathbf{h}}_u^H\hat{\mathbf{W}}_u\hat{\mathbf{h}}_u}|^2
\end{aligned}
\end{equation}
and
\begin{equation}
\begin{aligned}
\frac{1}{U}\tilde{\mathbf{w}}_u^H\tilde{\mathbf{w}}_u=& \lim_{\alpha\rightarrow 0} \frac{1}{U}\hat{\mathbf{h}}_u^H\hat{\mathbf{W}}^2\hat{\mathbf{h}}_u\\
=& \lim_{\alpha\rightarrow 0}\frac{\frac{1}{U}\hat{\mathbf{h}}_u^H\hat{\mathbf{W}}_u^2\hat{\mathbf{h}}_u}{(1+\frac{1}{U}\hat{\mathbf{h}}_u^H\hat{\mathbf{W}}_u\hat{\mathbf{h}}_u)^2}
\end{aligned}
\end{equation}
where $\hat{\mathbf{W}}=(\frac{1}{U}\hat{\mathbf{H}}\hat{\mathbf{H}}^H+\alpha \mathbf{I}_M)^{-1}$ and $\hat{\mathbf{W}}_u=(\frac{1}{U}\hat{\mathbf{H}}_u\hat{\mathbf{H}}_u^H+\alpha \mathbf{I}_M)^{-1}$.
Consequently, the equation \eqref{eq:signal} is reduced as
\begin{equation}
\label{eq:signal2}
\frac{1}{U}
\left|\hat{\mathbf h}_u^H
\mathbf w_u\right|^2=\lim_{\alpha\rightarrow 0}\frac{|\frac{\alpha}{U}\hat{\mathbf{h}}_u^H\hat{\mathbf{W}}_u\hat{\mathbf{h}}_u|^2}{\frac{\alpha^2}{U}\hat{\mathbf{h}}_u^H\hat{\mathbf{W}}_u^2\hat{\mathbf{h}}_u}
\end{equation}

For each BUE $u\in\mathcal U$, define
\begin{equation}
\hat{\boldsymbol{\Omega}}_u
=
\frac{\rho_u\mathbf R_u}
{1+\kappa_u},
\end{equation}
and
\begin{equation}
\hat{\mathbf{a}}_u
=
\frac{1}{\sqrt U}
\sqrt{
\frac{\rho_u\kappa_u}
{1+\kappa_u}
}
\mathbf h_u^{\rm LoS}.
\end{equation}
Collecting the deterministic LoS components of all BUEs, define
\begin{equation}
\hat{\mathbf A}
=
\left[
\hat{\mathbf a}_1,\ldots,\hat{\mathbf a}_U
\right].
\end{equation}

Applying Theorem~\ref{thm:zf_resolvent} to the full system with
$n=U$, $\{\hat{\boldsymbol{\Omega}}_i\}_{i=1}^{U}$, and $\mathbf A$,
let $\underline{\boldsymbol{\Theta}}$ denote the corresponding
ZF deterministic equivalent. Define
\begin{equation}
\underline{\delta}_u
=
\frac{1}{U}
\operatorname{tr}
\left(
\hat{\boldsymbol{\Omega}}_u
\underline{\boldsymbol{\Theta}}
\right),
\end{equation}
and
\begin{equation}
\underline{q}_u
=
\hat{\mathbf a}_u^H
\underline{\boldsymbol{\Theta}}
\hat{\mathbf a}_u.
\end{equation}

By Proposition~\ref{prop:LOO_quadratic_form_ZF}, we have
\begin{equation}
\lim_{\alpha\to0}
\frac{\alpha}{U}
\hat{\mathbf h}_u^H
\hat{\mathbf W}_u
\hat{\mathbf h}_u
-
\stackrel{\circ}{\mu}_u
\xrightarrow[N,U\to\infty]{a.s.}
0,
\end{equation}
where
\begin{equation}
\label{eq:mu_BUE_ZF}
\stackrel{\circ}{\mu}_u
=
\underline{\delta}_u
+
\frac{
\underline{\delta}_u
\underline{q}_u
}{
\underline{\delta}_u
-
\underline{q}_u
}.
\end{equation}
Equivalently,
\begin{equation}
\stackrel{\circ}{\mu}_u
=
\frac{
\underline{\delta}_u^2
}{
\underline{\delta}_u
-
\underline{q}_u
}.
\end{equation}

In addition, $\underline{\boldsymbol{\Theta}}$ is obtained from the
following fixed-point equations by applying
Theorem~\ref{thm:zf_resolvent} to the full system:
\begin{equation}
\label{eq:fixed_point_Theta}
\begin{cases}
\displaystyle
\underline{\delta}_{i}
=
\frac{1}{U}
\operatorname{tr}
\left(
\hat{\boldsymbol{\Omega}}_{i}
\underline{\boldsymbol{\Theta}}
\right),
&
i=1,\ldots,U,
\\[2mm]
\displaystyle
\underline{\widetilde{\delta}}_{i}
=
\left[
\underline{\widetilde{\boldsymbol{\Theta}}}
\right]_{i,i},
&
i=1,\ldots,U.
\end{cases}
\end{equation}
where
\begin{equation}
\underline{\mathbf F}
=
\left(
\mathbf I_N
+
\frac{1}{U}
\sum_{i=1}^{U}
\hat{\boldsymbol{\Omega}}_{i}
\underline{\widetilde{\delta}}_{i}
\right)^{-1},
\end{equation}

\begin{equation}
\underline{\widetilde{\mathbf F}}
=
\operatorname{diag}
\left(
\frac{1}{\underline{\delta}_{i}};
\, i=1,\ldots,U
\right),
\end{equation}

\begin{equation}
\underline{\boldsymbol{\Theta}}
=
\left(
\underline{\mathbf F}^{-1}
+
\hat{\mathbf A}
\underline{\widetilde{\mathbf F}}
\hat{\mathbf A}^H
\right)^{-1},
\end{equation}

and
\begin{equation}
\underline{\widetilde{\boldsymbol{\Theta}}}
=
\left(
\underline{\widetilde{\mathbf F}}^{-1}
+
\hat{\mathbf A}^H
\underline{\mathbf F}
\hat{\mathbf A}
\right)^{-1}.
\end{equation}

\begin{proposition}
\label{prop:ZF_limit_equivalence}
Let
\begin{equation}
\hat{\mathbf W}_u(\alpha)
=
\left(
\frac{1}{U}\hat{\mathbf H}_u\hat{\mathbf H}_u^H
+
\alpha\mathbf I_M
\right)^{-1},
\qquad \alpha>0.
\end{equation}
Then,
\begin{equation}
\lim_{\alpha\to0}
\frac{\alpha^2}{U}
\hat{\mathbf h}_u^H
\hat{\mathbf W}_u^2(\alpha)
\hat{\mathbf h}_u
=
\lim_{\alpha\to0}
\frac{\alpha}{U}
\hat{\mathbf h}_u^H
\hat{\mathbf W}_u(\alpha)
\hat{\mathbf h}_u .
\end{equation}
\end{proposition}

\begin{proof}
Consider the eigendecomposition
\begin{equation}
\frac{1}{U}
\hat{\mathbf H}_u
\hat{\mathbf H}_u^H
=
\mathbf U_u
\begin{bmatrix}
\boldsymbol{\Lambda}_u & \mathbf 0\\
\mathbf 0 & \mathbf 0
\end{bmatrix}
\mathbf U_u^H,
\end{equation}
where
$\boldsymbol{\Lambda}_u
=
\operatorname{diag}
(\lambda_1,\ldots,\lambda_r)$
contains the strictly positive eigenvalues. Partition
\begin{equation}
\mathbf U_u
=
\left[
\mathbf U_{u,\parallel},
\mathbf U_{u,\perp}
\right],
\end{equation}
where the columns of $\mathbf U_{u,\parallel}$ span
$\operatorname{range}(\hat{\mathbf H}_u)$, while the columns of
$\mathbf U_{u,\perp}$ span
$\operatorname{null}(\hat{\mathbf H}_u^H)$.

It follows that
\begin{equation}
\hat{\mathbf W}_u(\alpha)
=
\mathbf U_{u,\parallel}
(\boldsymbol{\Lambda}_u+\alpha\mathbf I)^{-1}
\mathbf U_{u,\parallel}^H
+
\frac{1}{\alpha}
\mathbf U_{u,\perp}
\mathbf U_{u,\perp}^H.
\end{equation}
Hence,
\begin{equation}
\lim_{\alpha\to0}
\alpha
\hat{\mathbf W}_u(\alpha)
=
\mathbf U_{u,\perp}
\mathbf U_{u,\perp}^H
\triangleq
\mathbf P_u^\perp.
\end{equation}

Similarly,
\begin{equation}
\hat{\mathbf W}_u^2(\alpha)
=
\mathbf U_{u,\parallel}
(\boldsymbol{\Lambda}_u+\alpha\mathbf I)^{-2}
\mathbf U_{u,\parallel}^H
+
\frac{1}{\alpha^2}
\mathbf U_{u,\perp}
\mathbf U_{u,\perp}^H,
\end{equation}
which gives
\begin{equation}
\lim_{\alpha\to0}
\alpha^2
\hat{\mathbf W}_u^2(\alpha)
=
\mathbf P_u^\perp.
\end{equation}

Therefore,
\begin{equation}
\begin{aligned}
\lim_{\alpha\to0}
\frac{\alpha}{U}
\hat{\mathbf h}_u^H
\hat{\mathbf W}_u(\alpha)
\hat{\mathbf h}_u
&=
\frac{1}{U}
\hat{\mathbf h}_u^H
\mathbf P_u^\perp
\hat{\mathbf h}_u
\\
&=
\lim_{\alpha\to0}
\frac{\alpha^2}{U}
\hat{\mathbf h}_u^H
\hat{\mathbf W}_u^2(\alpha)
\hat{\mathbf h}_u.
\end{aligned}
\end{equation}
This completes the proof.
\end{proof}

Substituting the above deterministic equivalents into
\eqref{eq:signal2} and using Proposition~\ref{prop:ZF_limit_equivalence},
we obtain
\begin{equation}
\label{eq:DE_desired_signal_BUE}
\frac{1}{U}
\left|
\hat{\mathbf h}_u^H\mathbf w_u
\right|^2
-
\stackrel{\circ}{\mu}_u
\xrightarrow[N,U\to\infty]{a.s.}
0.
\end{equation}

Furthermore, under Assumption~\ref{assump:HUE_interference}, the law of large
numbers yields
\begin{equation}
\label{eq:DE_interference_BUE}
\sum_{k\in\mathcal K}
p_k
\left|
\hat h_{ku}
\right|^2
-
\sum_{k\in\mathcal K}
p_k\rho_{ku}
\xrightarrow[K\to\infty]{a.s.}
0.
\end{equation}

Combining \eqref{eq:DE_desired_signal_BUE} and
\eqref{eq:DE_interference_BUE}, we obtain
\begin{equation}
\mathrm{SINR}_u^{(\rm BUE,DL)}
-
\overline{\mathrm{SINR}}_u^{(\rm BUE,DL)}
\xrightarrow[N,U,K\to\infty]{a.s.}
0,
\end{equation}
where
\begin{equation}
\overline{\mathrm{SINR}}_u^{(\rm BUE,DL)}
=
\frac{
P^{(\rm BS)}
\stackrel{\circ}{\mu}_u
}{
\displaystyle
\sum_{k\in\mathcal K}
p_k\rho_{ku}
+1
}.
\end{equation}
This completes the proof.

\bibliography{my_bibliography}
\bibliographystyle{IEEEtran}
\end{document}